\documentclass[sigconf, anonymous=false]{acmart}
\AtBeginDocument{%
  }

\usepackage{multirow}
\usepackage{array}
\usepackage{cancel}
\usepackage{tikz}
\usepackage{subcaption}
\usepackage{soul}
\setcopyright{acmlicensed}
\copyrightyear{2018}
\acmYear{2018}
\acmDOI{XXXXXXX.XXXXXXX}
\acmConference[Conference acronym 'XX]{Make sure to enter the correct
  conference title from your rights confirmation email}{June 03--05,
  2018}{Woodstock, NY}
\acmISBN{978-1-4503-XXXX-X/2018/06}

\begin{document}

\title{Setup Once, Secure Always: A Single-Setup Secure Federated Learning Aggregation Protocol with Forward and Backward Secrecy for Dynamic Users}
\author{Nazatul Haque Sultan}
\affiliation{%
  \institution{CSIRO's Data61, Australia}
  \country{}
}
\author{Yan Bo}
\affiliation{%
  \institution{CSIRO's Data61, Australia}
  \country{}
}
\author{Yansong Gao}
\affiliation{%
  \institution{UWA, Australia}
  \country{}
}
\author{Seyit Camtepe}
\affiliation{%
  \institution{CSIRO's Data61, Australia}
  \country{}
}
\author{Arash Mahboubi}
\affiliation{%
  \institution{CSU, Australia}
  \country{}
}

\author{Hang Thanh Bui}
\affiliation{%
  \institution{UNSW, Australia}
  \country{}
}
\author{Aufeef Chauhan}
\affiliation{%
  \institution{RMIT, Australia}
  \country{}
}

\author{Hamed Aboutorab}
\affiliation{%
  \institution{UNSW, Australia}
  \country{}
}

\author{Michael Bewong}
\affiliation{%
  \institution{CSU, Australia}
  \country{}
}
\author{Praveen Gauravaram}
\affiliation{%
  \institution{Research \& Innovation, TCS, Australia}
  \country{}
}
\author{Dineshkumar Singh}
\affiliation{%
  \institution{Research \& Innovation, TCS, India}
  \country{}
}
\author{Rafiqul Islam}
\affiliation{%
  \institution{CSU, Australia}
  \country{}
}
\author{Sharif Abuadbba}
\affiliation{%
  \institution{CSIRO's Data61, India}
  \country{}
}

\renewcommand{\shortauthors}{Trovato et al.}

\begin{abstract}
Federated Learning (FL) enables multiple users to collaboratively train a machine learning model without sharing raw data, making it suitable for privacy-sensitive applications. However, local model or weight updates can still leak sensitive information. Secure aggregation protocols mitigate this risk by ensuring that only the aggregated updates are revealed. Among these, single-setup protocols, where key generation and exchange occur only once, are the most efficient due to reduced communication and computation overhead. However, existing single-setup protocols often lack support for dynamic user participation and do not provide strong privacy guarantees such as forward and backward secrecy.
\par 
In this paper, we present a novel secure aggregation protocol that requires only a single setup for the entire FL training. Our protocol supports dynamic user participation, tolerates dropouts, and achieves both forward and backward secrecy. It leverages lightweight symmetric homomorphic encryption with a key negation technique to mask updates efficiently, eliminating the need for user-to-user communication. To defend against model inconsistency attacks, we introduce a low-overhead verification mechanism using message authentication codes (MACs). We provide formal security proofs under both semi-honest and malicious adversarial models and implement a full prototype. Experimental results show that our protocol reduces user-side computation by up to $99\%$ compared to state-of-the-art protocols like e-SeaFL (ACSAC’24), while maintaining competitive model accuracy. These features make our protocol highly practical for real-world FL deployments, especially on resource-constrained devices.

\end{abstract}



\maketitle

\section{Introduction}
\label{intro}
Recent research in Federated Learning (FL) has revealed serious privacy concerns. Studies such as \cite{Zhu2019, Geiping2020, zhou2023ppa, Jeon2024} demonstrate that adversaries, including a potentially malicious central server, can reveal private information of local datasets by analyzing the gradients exchanged during model updates. In response to these threats, secure aggregation has gained significant attention \cite{Chen2024}. This cryptographic technique safeguards user privacy by preventing the server from accessing individual model updates, thereby mitigating the risk of gradient-based reconstruction attacks \cite{Yin2021}.
\par 
Masking-based secure aggregation technique has recently garnered significant attention from the research community due to its efficiency, scalability, and strong privacy guarantees in federated learning settings \cite{Chen2024}. The core idea behind this technique is to conceal users' local model updates using cryptographic masks, which are canceled out during the aggregation phase at the server \cite{Bonawitz2017}. 
This ensures that the server can compute only the aggregated result without ever accessing any individual update. In their seminal work, Bonawitz et al. \cite{Bonawitz2017} introduced the first practical and scalable masking-based secure aggregation protocol, which tolerates user dropouts and minimizes user-side overhead. Building on this foundation, a range of protocols have been developed to enhance security and functionality while reducing communication and computation overhead \cite{Bell2020, Eltaras2023, Fereidooni2021, Guo2021, Liu2023, Xu2020}.
\par 
However, Bonawitz et al.~\cite{Bonawitz2017} protocol requires a setup operation in each training round, during which users establish pairwise cryptographic keys. This setup phase is similarly adopted by many subsequent protocols~\cite{Bell2020, Eltaras2023, Fereidooni2021, Guo2021, Liu2023, Xu2020}. While this design enhances resilience to user dropout and preserves privacy, it introduces additional communication and computation overhead. As noted in~\cite{Ma2023}, these costs become particularly significant in large-scale FL deployments, rendering such approaches less practical for real-world applications. 
\par 
To address the limitations of Bonawitz et al.~\cite{Bonawitz2017} and its variants~\cite{Bell2020, Eltaras2023, Fereidooni2021, Guo2021, Liu2023, Xu2020}, Ma et al.~\cite{Ma2023} (S\&P'23) and Behnia et al.~\cite{Behnia2024} (ACSAC'24) proposed two masking-based secure aggregation protocols, named Flamingo and e-SeaFL, respectively. Both protocols utilize a single setup operation that supports multiple training rounds in FL. This design simplifies the training process by reducing communication and computation costs, thereby accelerating overall convergence. The core idea is to pre-establish pairwise secrets or seed materials between users in Flamingo \cite{Ma2023} and between users and assistant nodes in e-SeaFL \cite{Behnia2024} during the initial setup. These \emph{long-term} seeds or secrets are then used to deterministically derive per-round masking values using lightweight cryptographic primitives such as pseudo-random functions (PRFs), effectively eliminating the need for costly setup phases in each round.
\par 
\par 
However, several challenges remain with Flamingo and e-SeaFL. Most critically, these protocols lack forward and backward secrecy. Although they eliminate the need for repeated setup in each round by using long-term pairwise secrets, this design introduces a vulnerability: if a long-term secret key is compromised, an adversary could reconstruct all past and future masking values for that user within the same training session. This concern is particularly relevant in FL settings, where training can span days and devices may operate in untrusted or adversarial environments. 
Consequently, forward and backward secrecy remain a crucial yet unresolved requirement in these single-setup aggregation protocols. We provide a detailed analysis of forward and backward secrecy for e-SeaFL and Flamingo in Appendix~\ref{appendix:fs-bs-analysis}. Second, these protocols do not support dynamic participation; users must be registered before the initial setup phase and cannot join the training session once it has started. This limitation reduces their practicality in real-world FL deployments, where user availability is often unpredictable. Finally, handling user dropout remains inefficient with Flamingo. It can tolerate up to one-third of corrupted users, including both regular and decryptor users.
 \par 
In this paper, we introduce a lightweight masking protocol for FL based on additive symmetric homomorphic encryption (HE) with key negation. Like Flamingo and e-SeaFL, our protocol requires only a single setup operation for all rest FL training rounds, significantly reducing communication overhead. 
However, our design offers key improvements in both security and efficiency. Unlike Flamingo and e-SeaFL, which rely on PRF operations using long-term secret seeds to generate masking values for each element of the model vector, our protocol uses simple random secrets directly as masking parameters.
This eliminates the need for per-element PRF evaluations, resulting in substantially lower computational cost and making our approach highly efficient and practical--- especially for resource-constrained devices. 
To further enhance scalability and robustness, our protocol introduces a small set of intermediate servers or edge servers (which may be selected users or assistant nodes similar to Flamingo and e-SeaFL, respectively) to assist in partial aggregation. 
These intermediate servers collect masked model updates from users and support the central server, referred to as the Aggregator, in managing user dropouts. 
This intermediate layer provides three key advantages:
\begin{itemize} 
\item First, it eliminates the need for user-to-user interaction unlike the Bonawitz et al.~\cite{Bonawitz2017} and its variants~\cite{Bell2020, Eltaras2023, Fereidooni2021, Guo2021, Liu2023, Xu2020}, which would otherwise impose a significant burden on resource-constrained users~\cite{Liu2023}. 

\item Second, avoiding direct user interactions reduces the trust assumptions required in existing protocols, such as those proposed by Bonawitz et al.~\cite{Bonawitz2017} and its variants~\cite{Bell2020, Eltaras2023, Fereidooni2021, Guo2021, Liu2023, Xu2020}, for maintaining privacy. This design choice also mitigates the risk of user collusion. Our protocol maintains privacy guarantees as long as at least two users behave honestly, a minimal assumption in FL, ensuring that their individual model updates remain protected. 

\item Third, intermediate servers in hierarchical FL architectures offer a practical and scalable solution for handling aggregation and privacy-preserving operations. As demonstrated by Ratnayake et al.~\cite{RATNAYAKE2025128483}, these servers are typically deployed on infrastructure with sufficient computational capacity and reliable connectivity, such as edge or fog nodes, making them well-suited for tasks like model aggregation, homomorphic encryption, and differential privacy noise addition.
\end{itemize}

Notably, the inclusion of intermediate servers in our protocol aligns with the trust model adopted in several recent state-of-the-art works~\cite{Wang2018, kairouz2022, Ma2023}, where helper nodes or designated users assist the aggregator during the aggregation phase. More recently, Wang et al.~\cite{Wang2025} introduced a similar approach involving edge or intermediate servers. Consistent with these models, our protocol remains robust against collusion attacks, provided that at least one intermediate server behaves honestly. Overall, our approach enhances both performance and security while minimizing the trust assumptions placed on users, without increasing the trust requirements for intermediate servers.

\section*{Contributions} Our work makes the following key contributions:
\begin{itemize}
    \item We propose a new secure aggregation protocol for FL that uses lightweight symmetric homomorphic encryption with a key negation technique. It requires only one setup for all training rounds and avoids costly per-element PRF operations, making it more efficient than existing state-of-the-art protocols like Flamingo and e-SeaFL.
    \item Our protocol is the first single-setup secure aggregation protocol to support dynamic user participation, allowing users to join or leave at any round without reinitialization or user-to-user communication--- an important feature missing in existing single-setup protocols like Flamingo and e-SeaFL. It is also the first single-setup secure aggregation protocol to achieve both forward and backward secrecy, ensuring that compromise of long-term keys does not reveal past or future model updates. 
    \item We introduce a lightweight model verification method to detect model inconsistency attacks using Message Authentication Codes (MACs) and hash functions. This approach avoids expensive cryptographic operations and adds an essential layer of protection for secure aggregation protocols in FL \cite{Pasquini2022}.
    \item We provide formal security proofs under both semi-honest and malicious adversarial models, demonstrating robustness against collusion and active attacks. We also present a comprehensive comparison with the most efficient state-of-the-art protocol e-SeaFL, showing that our protocol achieves better performance in terms of computation, communication, security, and functionality. We implement a fully end-to-end prototype and release open-source code with extensive experimental results: \url{https://anonymous.4open.science/r/scatesfl-44E1} (performance testing) and \url{https://anonymous.4open.science/r/scatesfl-815A} (real-world accuracy evaluation).
\end{itemize}

\begin{table*}[!t]
\centering
\tiny
\caption{Comparison of Functionality and Security}\label{table:functionality-comparison}
\begin{tabular}{|l|l|l|l|l|l|l|l|}
\hline
 & Threat Model & Forward \& Backward Secrecy & Communication Round & Setup Round(s) & User Dropout& Model Verification & Dynamic User\\ \hline
SecAgg \cite{Bonawitz2017} & Malicious & yes & 4 & $|i|$ & $\frac{m}{3}$ & no & no\\ \hline
SecAgg+ \cite{Bell2020} & Malicious & yes & 3 & $|i|$ & $\sigma\cdot m$ &no & no\\ \hline
FastSecAgg \cite{Kadhe2020} & Semi Honest & yes & 3 & $|i|$& $\frac{m}{2}- 1$ &no & no\\ \hline
EDRAgg \cite{Liu2023}  & Malicious & yes & 3 & $|i|$ & $\frac{m}{2}- 1$ & no & no\\ \hline
Flamingo \cite{Ma2023}   & Malicious & no & 3 & 1 & $\frac{\delta_d+\eta_d}{3}$& yes & no\\ \hline
e-SeaFL \cite{Behnia2024} & Malicious & no & 2 & 1 & $m-2$& yes & no\\ \hline 
Ours & Malicious & yes & 2 & 1 & $m-2$& yes & yes\\ \hline
\end{tabular}
\\\scriptsize{\vspace{.1cm}$m$ represents the total number of participating users; $\sigma$ represents a security parameter and $\sigma\cdot m\leq \frac{m}{3}$; $|i|$ represents the total number of rounds in the training phase; $\eta_d$ represents the number of corrupted decryptors; $\delta_d$ represents decryptor dropout rate.}
\end{table*}

\begin{table*}[!t]
\centering
\caption{Asymptotic Computation and Communication Overhead Comparison}
\label{table:fun-com-commu-comparison}
\tiny
\begin{tabular}{|l|ccc|ccc|}
\hline
\multirow{2}{*}{Protocols} & \multicolumn{3}{c|}{Computation Overhead} & \multicolumn{3}{c|}{Communication Overhead} \\ \cline{2-7} 
 & \multicolumn{1}{l|}{User} & \multicolumn{1}{l|}{Intermediate Server} & Aggregator & \multicolumn{1}{l|}{User} & \multicolumn{1}{l|}{Intermediate Server} & Aggregator  \\ \hline
   SecAgg \cite{Bonawitz2017} &\multicolumn{1}{c|}{$\mathcal{O}(m^2+ m\cdot |v|)$}&\multicolumn{1}{c|}{NA}&\multicolumn{1}{c|}{$\mathcal{O}(m^2+ m\cdot |v|)$}&\multicolumn{1}{c|}{$\mathcal{O}(m+ |v|)$}&\multicolumn{1}{c|}{NA}&\multicolumn{1}{c|}{$\mathcal{O}(m^2+ m\cdot |v|)$}\\\hline
SecAgg+ \cite{Bell2020} & \multicolumn{1}{c|}{$\mathcal{O}(|v|\cdot\log m+ \log m)$} & \multicolumn{1}{c|}{NA} & $\mathcal{O}(m\cdot |v|+ m\log m)$ & \multicolumn{1}{c|}{$\mathcal{O}(|v|+ \log m)$} & \multicolumn{1}{c|}{NA} & $\mathcal{O}(m\cdot |v|+ m\log m)$ \\ \hline

FastSecAgg \cite{Kadhe2020} & \multicolumn{1}{c|}{$\mathcal{O}(|v|\log m)$} & \multicolumn{1}{c|}{NA} & $\mathcal{O}(|v|\cdot \log m)$ & \multicolumn{1}{c|}{$\mathcal{O}(m+|v|)$} & \multicolumn{1}{c|}{NA} & $\mathcal{O}(m^2+ m\cdot |v|)$ \\ \hline
EDRAgg \cite{Liu2023} & \multicolumn{1}{c|}{$\mathcal{O}(m^2+ |v|)$} & \multicolumn{1}{c|}{NA} & $\mathcal{O}(m+ |v|)$ & \multicolumn{1}{c|}{$\mathcal{O}(m+ |v|)$} & \multicolumn{1}{c|}{NA} & $\mathcal{O}(m^2+ m\cdot |v|)$ \\ \hline

Flamingo$^*$ \cite{Ma2023} & \multicolumn{1}{c|}{$\mathcal{O}(S\cdot |v|+ d\cdot |v|)$} & \multicolumn{1}{c|}{$\mathcal{O}(d^2+ \delta\cdot S\cdot m+ (1- \delta)m+ \epsilon\cdot m^2)$} & $\mathcal{O}(m|v|+ m^2)$ & \multicolumn{1}{c|}{$\mathcal{O}(m+ |v|)$} & \multicolumn{1}{c|}{$\mathcal{O}(d+ \delta\cdot S\cdot m+ (1- \delta)m)$} & $\mathcal{O}(m\cdot |v|+ m\cdot d+ m\cdot S))$ \\ \hline
e-SeaFL$^*$ \cite{Behnia2024} & \multicolumn{1}{c|}{$\mathcal{O}(d\cdot|v|+ |v|)$} & \multicolumn{1}{c|}{$\mathcal{O}(m\cdot |v|+ m)$} & $\mathcal{O}(m\cdot |v|+ d\cdot |v|)$ & \multicolumn{1}{c|}{$\mathcal{O}(d\cdot |v|)$} & \multicolumn{1}{c|}{$\mathcal{O}(m\cdot |v|)$} & $\mathcal{O}(m\cdot |v|+ d\cdot |v|)$\\ \hline
Our protocol & \multicolumn{1}{c|}{{\bf$\mathcal{O}(d\cdot |v|+ |v|)$}} & \multicolumn{1}{c|}{$\mathcal{O}(m\cdot |v|)$} & $\mathcal{O}(m\cdot |v|+ d\cdot |v|)$ & \multicolumn{1}{c|}{$\mathcal{O}(d\cdot |v|)$} & \multicolumn{1}{c|}{$\mathcal{O}(m\cdot |v|)$} & $\mathcal{O}(m\cdot |v|+ d\cdot |v|)$ \\ \hline
\end{tabular}
\\\scriptsize{\vspace{.1cm}$^*$The set of decryptors in Flamingo and assistant nodes in e-SeaFL represented as the intermediate servers for comparison; $m$ represents total number of participating users; $|v|$ represents the size of the input vector; NA: Not applicable; $\sigma$ represents a security parameter and $\sigma\cdot m\leq \frac{m}{3}$; $|i|$ represents the total number of rounds in the training phase; $d$ represents the number of decryptors or intermediate servers; $\delta$ represents user dropout rate; $S$ represents the upper bound on the number of neighbors of a user; $\eta_d$ represents the number of corrupted decryptors; $\delta_d$ represents decryptor dropout rate; $\epsilon$ represents the graph generation parameter in \cite{Ma2023}.}
\end{table*}

\section{Related Work}
\label{sec:related-work}

Secure aggregation in FL protects individual users' local model parameters from being disclosed to the central aggregator server. This mechanism is also referred to as privacy-preserving aggregation \cite{Liu2022}. Differential privacy (DP), homomorphic encryption (HE), secure multi-party computation (SMC), and masking methods are being used for privacy-preserving secure aggregation in FL \cite{Yin2021}. Our proposed protocol falls within the masking-based method. Therefore, we briefly introduce the other secure aggregation methods while discussing masking-based approach.  
\subsection{DP, HE, and SMC Based Protocols}
In DP-based secure aggregation protocols, random noise is added to users' gradients before they are sent to the aggregator server \cite{Wei2020}. This process protects the sensitive information within the gradients by making it difficult to reverse-engineer the original values, as DP introduces controlled noise to the data. Despite this added layer of privacy, the server can still aggregate these perturbed gradients to approximate the true model updates, owing to the mathematical properties of DP. However, a key challenge with DP is the inevitable trade-off between the level of privacy and the utility of the data. The more noise that is introduced to protect privacy, the less accurate the aggregated model becomes, which can potentially lead to significant degradation in model performance \cite{Liu2023}. 
Striking a delicate balance between safeguarding individual privacy and maintaining data utility is a central concern in the application of DP in FL, and finding the optimal point on this spectrum remains an area of active research. Some notable works in this domain include \cite{Wei2020, Wang2020, Zhou2022}.
\par 
In HE-based aggregation protocols, such as those proposed in \cite{Chan2012, Phong2018, Ma2022, Sav2021, Zhang2020Usenix} users first apply computation-intensive algorithms to encrypt their local model parameters before sending them to the central server.
HE allows for arithmetic operations to be performed directly on encrypted data without requiring it to be decrypted first. This means that the server can aggregate the encrypted local model parameters by performing operations like addition or multiplication directly on the ciphertext. After the aggregation process, the server either sends the aggregated encrypted result back to the users for decryption or continues the training process directly on the ciphertext. This approach ensures that sensitive gradient information remains protected throughout the entire process, as the data remains encrypted during both transmission and computation. However, the use of HE introduces significant computational overhead, both in terms of the initial encryption performed by the participants and the subsequent operations carried out by the server on the encrypted data \cite{Liu2023}, \cite{Liu2022}. 

\par 
SMC is another cryptographic technique that has been applied in secure aggregation within FL \cite{So2021, Kadhe2020, So2022ML}. SMC allows multiple participants, each with their private data, to collaboratively compute a desired objective function without revealing their data to others. This method ensures that each participant’s data remains confidential while still enabling the computation of an accurate result \cite{Yin2021}. However, SMC-based protocols tend to be inefficient, as they often involve significant communication overhead and face challenges in managing user dropouts \cite{Liu2022}, \cite{Ma2023}.

\subsection{Masking-Based Protocols}
The core idea of masking-based protocols is to secure users' local model parameters by adding random values (often referred to as one-time pads) to them before they are sent to the aggregator server \cite{Bonawitz2017}. These random values are designed so that they cancel out during the aggregation process, allowing the aggregator server to recover the aggregated plaintext gradient values for all participants in that round. The primary goal of the masking-based approach is to protect individual users' local model parameters from being exposed to unauthorized parties while still enabling the aggregator server to access the aggregated gradient in its plaintext form \cite{Liu2022}. Our work aligns with this approach.
\par 
In \cite{Bonawitz2017}, Bonawitz et al. proposed a practical secure aggregation protocol that uses masking to conceal individual users' local model parameters from the aggregator server. The protocol also employs secret sharing techniques to accommodate user dropouts, ensuring that the learning process remains unaffected. Building on the core concepts introduced by Bonawitz et al., several subsequent works have been proposed over the years to enhance security and functionality, including \cite{Xu2020, Guo2021, Fereidooni2021, Liu2023, Bell2020, Eltaras2023}. In \cite{Xu2020} and \cite{Guo2021}, the authors proposed two protocols to add verifiability on top of the protocol \cite{Bonawitz2017}. In \cite{Fereidooni2021}, Fereidooni et al. used secret sharing and homomorphic encryption to achieve secure aggregation without relying on a trusted third party to generate any public/private key pair for the clients. In \cite{Liu2023}, Liu et al. applied the concepts of homomorphic pseudo-random generator and Shamir Secret Sharing technique in achieving user dropouts and reducing the communication costs among the users, which is secure against both semi-honest and malicious adversaries. In \cite{Bell2020}, the authors reduced the communication overhead of \cite{Bonawitz2017} by utilizing a logarithmic degree k-regular graph. In \cite{Liu2024}, Liu et al. proposed a Dynamic User Clustering protocol that builds on existing masking-based secure aggregation protocols, such as \cite{Bonawitz2017}, and incorporates a sparsification technique to address the interoperability issues with sparsification. In \cite{Fu2024}, Fu et al. proposed a blockchain-based decentralized secure aggregation protocol to replace the central aggregator server. The protocol uses masking and Shamir's Secret Sharing to ensure privacy. In \cite{Eltaras2023}, Eltaras et al. proposed a pairwise masking-based secure aggregation protocol that uses auxiliary nodes to achieve verifiability and handle user dropouts. 
In \cite{Yang2023}, Yang et al. proposed a single mask secure aggregation protocol for FL by combining the concepts of homomorphic Pseudorandom Generator, homomorphic Shamir secret sharing, and Paillier encryption. However, all these protocols (\cite{Bonawitz2017, Xu2020, Guo2021, Fereidooni2021, Liu2023, Bell2020, Yang2023, Eltaras2023, Liu2024, Fu2024}) work well for one round of training, they become highly inefficient when multiple training rounds are required (which is essential for most real-world FL applications) due to the need for an expensive setup phase that involves four communication rounds to establish shared randomness and pairwise keys in every FL training round \cite{Ma2023}. 
\par 
In \cite{Wang2023}, Wang et al. proposed a single-mask secure aggregation protocol based on the Decisional Composite Residuosity (DCR) assumption, incorporating non-interactive zero-knowledge (NIZK) proofs to enable result verification. However, as noted by Wu et al. in \cite{Wu2024}, the protocol is vulnerable to integrity attacks, where local gradients or aggregation results, along with their corresponding authentication tags or proofs, can be tampered without being detected by the verifiers. In \cite{Fazli2023}, Khojir et al. introduced a secure aggregation protocol based on additive secret sharing. The protocol utilizes a three-layered architecture (i.e., clients, middle servers, and lead server), which reduces communication costs for users compared to other protocols that rely on masking processes using secret sharing, which requires distributing secret shares among users, such as those in \cite{Xu2020, Guo2021, Fereidooni2021, Liu2023, Bell2020, Eltaras2023}. However, it is unable to handle user dropouts, which is essential for real-world applications. In \cite{Guo2024}, Guo \emph{et al.} proposed MicroSecAgg, a secure aggregation protocol designed to improve efficiency in single-server federated learning. However, Zhang \emph{et al.} \cite{Zhang2024b} later identified a privacy vulnerability in MicroSecAgg, demonstrating that predictable masking values could be exploited to compromise user privacy.
\par 
Recently, in \cite{Ma2023} and \cite{Behnia2024}, Ma et al. and Behnia et al. proposed ``Flamingo” and ``e-SeaFL”, respectively-- multi-round, single-server secure aggregation protocols that do not require an initialization setup phase for each training round (a single setup phase suffices). These protocols significantly reduce computation and communication overhead compared to repeated setup-based approaches such as \cite{Bonawitz2017, Xu2020, Guo2021, Fereidooni2021, Liu2023, Bell2020, Liu2024, Fu2024}. However, as discussed in our security analysis in Appendix \ref{appendix:fs-bs-analysis}, both Flamingo and e-SeaFL lack support for both forward and backward secrecy, which are important for single-setup secure aggregation protocols. Furthermore, they are not designed to operate in dynamic environments where new users may join or drop the training session at different times. Both protocols require prior knowledge of all participating users before the initial setup phase, due to the need for pairwise key generation. 
\par 
Table~\ref{table:functionality-comparison} and Table~\ref{table:fun-com-commu-comparison} present a comparative analysis of our protocol against prominent masking-based secure aggregation protocols. Earlier protocols such as SecAgg~\cite{Bonawitz2017}, SecAgg+~\cite{Bell2020}, and FastSecAgg~\cite{Kadhe2020} support forward and backward secrecy but require multiple setup rounds and lack support for dynamic user participation and model verification\footnote{Model verification refers to the process of ensuring the integrity, authenticity, and correctness of the aggregated global model.}. More recent protocols, including Flamingo~\cite{Ma2023} and e-SeaFL~\cite{Behnia2024}, reduce setup complexity and introduce model verification capabilities, yet they do not offer forward/backward secrecy or support for dynamic users. In contrast, our protocol is the only single-setup protocol that simultaneously provides forward and backward secrecy, dynamic user participation, and model verification, while maintaining low communication and computational overhead as shown in Table \ref{table:fun-com-commu-comparison}\footnote{In terms of asymptotic computation and communication overhead, our protocol is comparable to e-SeaFL~\cite{Behnia2024}. However, in practice, it achieves better computational efficiency, as e-SeaFL relies on relatively expensive elliptic curve operations and requires more frequent PRF evaluations. In contrast, our protocol leverages lightweight modular arithmetic. Detailed empirical performance results are provided in Section~\ref{sec:performance_eval}.}.

\section{Forward and Backward Secrecy in Single-Setup Secure Aggregation}
\label{sec:forward-backward-secrecy-summary}
Forward secrecy (FS) and backward secrecy (BS) are essential privacy guarantees in secure aggregation based FL systems, especially in long-running deployments. While single-setup secure aggregation protocols like e-SeaFL~\cite{Behnia2024} and Flamingo~\cite{Ma2023} offer efficiency benefits, they lack FS and BS due to their reliance on static long-term secrets. This limitation poses a significant privacy risk: if a long-term secret is compromised, an adversary can reconstruct all past and future model updates in that training session. We provide a detailed case study of this vulnerability in Appendix~\ref{appendix:fs-bs-analysis}, including the mask derivation process and practical implications.
 
\section{Preliminaries}
\label{sec:preli}
In this section, we briefly introduce the core concepts of symmetric homomorphic encryption and digital signatures, which are integral to the design of our protocol. 

\subsection{Symmetric Homomorphic Encryption}
\label{sec:symmetricKeyHomoEncrypt}
Our protocol employs a key negation method to mask the user's local model parameters. This method is based on the symmetric homomorphic encryption mechanism proposed in \cite{Castelluccia2009}. In this section, we will briefly describe the work of \cite{Castelluccia2009}. The key negation method will be explained in Section \ref{sec:tech_intuition}.
\par 
Let $m_i \in \mathbb{Z}_q$ represent a secret message, where $q$ is a large public prime. The message $m$ can be encrypted as follows:
\begin{align}
c_i=& \mathsf{Enc_{k_i}}(m_i)= m_i+ \mathtt{k_i} \mod{q}.
\end{align} 
where $k_i\in \mathbb{Z}_q$.
A receiver of $c_i$ with the given secret key $\mathtt{k_i}$ can recover the secret message $m_i$ as follows:
\begin{align}
m_i=& \mathsf{Dec_{k_i}}(c_i)= c_i- \mathtt{k_i}\mod{q}.
\end{align}
\begin{align}\label{eq:add_cipher}
    \mathsf{Dec_{(k_i+ k_j)}}(c_i+ c_j)=& \mathsf{Dec_{k_i}}(c_i)+ \mathsf{Dec_{k_j}}(c_j)\notag\\
    =& c_i- \mathtt{k_i}+ c_j- \mathtt{k_j}\mod{q}\notag\\
    =& m_i +m_j +[(\cancel{\mathtt{k_i}+ \mathtt{k_j}}) -(\cancel{\mathtt{k_i}+ \mathtt{k_j}})]\mod{q}\notag\\
    =& m_i+ m_j.
\end{align}
The protocol described above has additive homomorphic properties \cite{Castelluccia2009}. The addition of two ciphertexts is illustrated in Equation \ref{eq:add_cipher}. This property can also be extended to aggregate a set of ciphertexts using the corresponding aggregated secret keys, as shown in Equation \ref{eq:agg_cipher}.
\begin{align}\label{eq:agg_cipher}
    \sum_{1\leq i\leq n}m_i= & \mathsf{Dec}_{\sum_{1\leq i\leq n}k_i}\left(\sum_{1\leq i\leq n}c_i\right)\\
    =& \sum_{1\leq i\leq n}c_i - \sum_{1\leq i\leq n}k_i \mod{q}.
\end{align}
The described symmetric homomorphic encryption scheme is semantically secure (IND-CPA) if the secret keys $\mathtt{k_i}$, where $i\in\{1, 2, \cdots, n\}$, are generated randomly and no keys are reused. Please refer to \cite{Castelluccia2009} for detailed security proof.

\subsection{Signature Scheme}
Our protocol employs digital signatures to ensure the integrity and authenticity of messages exchanged between entities, leveraging a UF-CMA secure scheme as described in \cite{Bonawitz2017}. In general, a digital signature scheme consists of the following probabilistic polynomial-time (PPT) algorithms:
\begin{itemize}
    \item $(\mathsf{priv_{u_i}, pub_{u_i}}) \leftarrow \mathsf{SIG.Gen}(\lambda)$: This algorithm takes a security parameter $\lambda$ as input and outputs a pair of private and public keys $(\mathsf{priv_{u_i}, pub_{u_i}})$ for an entity, denoted as $\mathsf{u_i}$.
    \item $\sigma \leftarrow \mathsf{SIG.Sign} (\mathsf{priv_{u_i}, m})$: This algorithm takes as input the private key $\mathsf{priv_{u_i}}$ and a message $\mathsf{m}$, and outputs a signature for the message $\mathsf{m}$.
    \item $\{0, 1\} \leftarrow \mathsf{SIG.Ver} (\mathsf{pub_{u_i}}, \mathsf{m}, \sigma)$: This algorithm takes the public key $\mathsf{pub_{u_i}}$, a message $\mathsf{m}$, and the signature $\sigma$ as input, and outputs a bit indicating whether the signature is valid or invalid.
\end{itemize}

\begin{figure}[!t]
    \centering
	\fbox{\scalebox{5}{\includegraphics[width=1cm, height=.7cm]{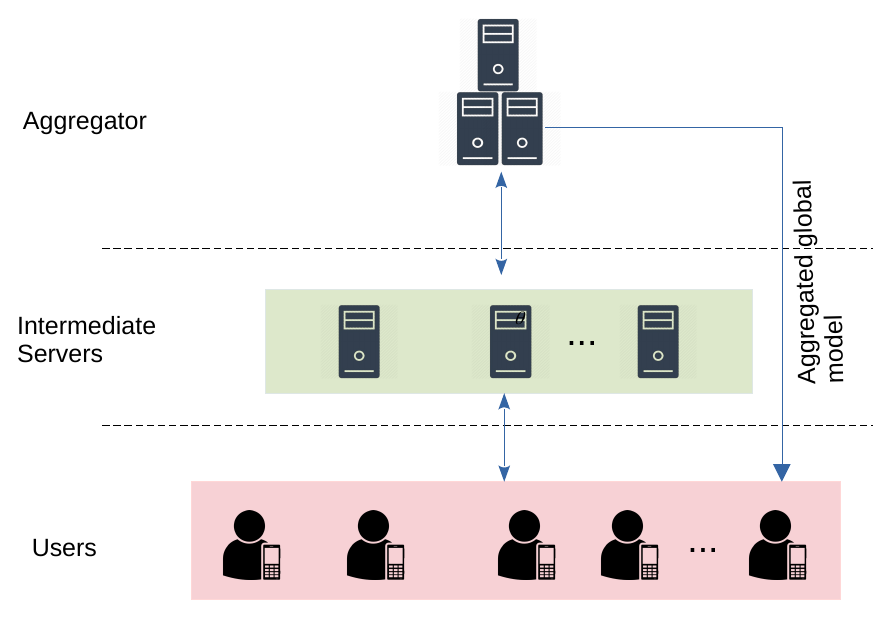}}}
	\caption{System Architecture}		
    \label{fig:system-architecture}
 \end{figure} 
 
\section{System Architecture, Threat Model, and Assumptions of Our Proposed Protocol}
\label{sec:architecture-threat-model}
In this section, we present the system architecture, threat model, and security assumptions of our proposed protocol. We first start with the system architecture.

\subsection{System Architecture}
\label{sec:system-architecture}
Figure \ref{fig:system-architecture} illustrates the system architecture of our protocol, which comprises three primary entities: users, intermediate servers, and the central server which we termed as Aggregator. 
\paragraph*{Users} They are the edge devices, such as smartphones and IoT devices, for \emph{cross-device FL}\footnote{Cross-device FL refers to a distributed machine learning method where many devices, such as smartphones or IoT devices, train a shared global model without sharing their local data.}, or organizations like hospitals, banks, universities, and government agencies for \emph{cross-silo FL}\footnote{Cross-silo FL involves a small number of organizations, such as hospitals or businesses, collaboratively training a model without sharing their private data. Unlike cross-device FL, cross-silo FL usually deals with higher-quality data from trusted institutions and is designed for scenarios with stricter privacy and security requirements.}. These users generate and own data locally. They perform local model training on their own data, downloading the current global model from the aggregator, training it, and then sending the updated model parameters to the intermediate servers and aggregator. Before transmission, the parameters are masked, ensuring that only the aggregator can recover the final aggregated model parameter, thereby maintaining the confidentiality of individual users' local model data. 
\paragraph*{Intermediate servers} They act as aggregation points between users and the main aggregator. They are to reduce communication overhead by combining model updates from multiple users before sending them to the aggregator. These servers receive model updates from users, aggregate the updates, and then forward the aggregated results to the aggregator for final processing. Additionally, they can help users detect model inconsistency attacks originating from the aggregator. In practical applications, fog nodes, edge servers, and third-party cloud services can function as intermediate servers, especially in edge computing or distributed system scenarios. By facilitating localized processing, these components reduce latency and bandwidth usage.
\paragraph*{Aggregator} It is the central server that handles the entire FL process. It maintains the global model, coordinates the FL training rounds, and aggregates model updates directly from the intermediate servers and users. The aggregator initiates the training process by distributing the initial global model to users. After receiving aggregated updates from the intermediate servers, the aggregator combines these updates to improve the global model. It then sends the updated global model back to users for the next round of training. 

\subsection{Threat Model}
\label{sec:threat-model}
Our threat model considers two types of adversaries: semi-honest and malicious. 
\par In the semi-honest model, we assume the aggregator is honest in executing assigned tasks but attempts to learn individual users' local training models to infer their datasets. Similarly, the intermediate servers are also considered semi-honest. Our primary goal in this model is to protect the confidentiality and privacy of each user's local training model. 

\par In the malicious model, we assume a more challenging scenario where both the aggregator and the intermediate servers may act maliciously, attempting to compromise the individual users' local training models. In this scenario, we assume that the malicious aggregator and the malicious intermediate servers can collude with up to $m-2$ users, where $m$ is the total number of participating users in an FL round. Similarly, we assume that the malicious aggregator can collude with up to $n'-1$ malicious intermediate servers in an FL round, where $n'$ is the total number of participating intermediate servers. Additionally, our threat model considers model inconsistency attacks by the malicious aggregator, where the malicious aggregator provides different parameters to different users to exploit behavioral differences in the model updates, thereby inferring information on users' datasets \cite{Pasquini2022}. We aim to detect such behavior by the malicious aggregator.
\par 
We also consider the possibility of malicious users in our threat model. These users can collude with other users as well as with the aggregator and the intermediate servers to gain knowledge of other users' local training models. However, we assume that up to $m-2$ users can collude, meaning at least two users must remain honest at all times. Without this assumption, the users could collude with the aggregator to obtain the local training model update of a targeted user.

\subsection{Assumption}
\label{sec:assumption}
We make some assumptions in designing our protocol. First, we assume that the aggregator and intermediate servers remain online throughout the training process. Second, each entity possesses a public-private key pair managed via a Public Key Infrastructure (PKI). Third, all users are uniquely identifiable. Fourth, authenticated and private communication channels exist between users and intermediate servers, users and the aggregator, and between intermediate servers and the aggregator.

\begin{figure}[h]
    \centering
	\fbox{\scalebox{5}{\includegraphics[width=1.2cm, height=.4cm]{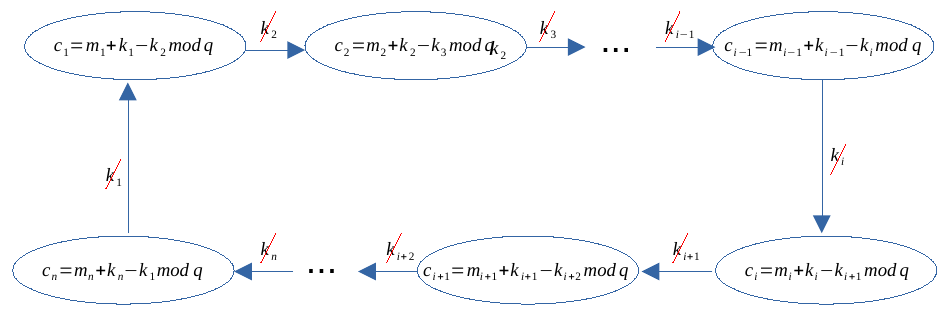}}}
	\caption{Sample Key Negation Mechanism}		
    \label{fig:key-neg}
 \end{figure}

\begin{figure*}[htbp]
  \centering
\begin{scriptsize}
  \begin{tikzpicture}[
    protocolbox/.style={draw, rectangle, rounded corners, inner sep=10pt},
    title/.style={font=\bfseries}
  ]
    \node[protocolbox] (protocol) {
      \begin{minipage}{.9\textwidth}
        \centering \textbf{Our Proposed Secure Aggregation Protocol}
        \begin{itemize}
            \item \textbf{Setup}
                \begin{itemize}
                \item All the entities agree on the security parameter $\lambda$, a hash function $H: \{0, 1\}^*\rightarrow \{0, 1\}^l$, a set $\mathbb{Z}_q$ of integer modulo $q$, and a threshold value $t$.  
                \item Each user, intermediate server, and aggregator are assigned unique identifiers (denoted as $u_i$ for the $i^{th}$ user, $f_i$ for the $i^{th}$ intermediate server, and $\mathrm{Agg}$ for the aggregator). 
                \item \textcolor{red}{Each user, intermediate server, and aggregator are assigned a public and private key pair $(\mathsf{pub_{u_i}}, \mathsf{priv_{u_i}}), (\mathsf{pub_{f_i}}, \mathsf{priv_{f_i}})$, and $(\mathsf{pub_{Agg}}, \mathsf{priv_{Agg}})$, respectively, for signature generation.}
                \item All users know the participating intermediate servers and the aggregator set $\mathcal{N}$, where $|\mathcal{N}|= n$.
                \item All intermediate servers and the aggregator have the universal set of the users $\mathcal{M}$, where a random subset $\mathcal{U}$ ($m= |\mathcal{U}|$) of users participate in each round of training.
               \item All intermediate servers communicate with the aggregator over private, authenticated channels.
                \item All users have private, authenticated channels with both the intermediate servers and the aggregator.
               \end{itemize}
            \item \textbf{Masking Round}
               \begin{itemize}            
                \item Users:
                \begin{itemize}
                    \item Each user $u_i$ chooses random masks (i.e., secret vectors) $\{k_j\}_{\forall j\in \mathcal{N}}\in \mathbb{Z}_q$.
                    \item Each user $u_i$ generates a directed cycle for the participating entities in $\mathcal{N}$, as described in Section \ref{sec:tech_intuition}.
                    \item For each $j^{th}$ node in the directed cycle, each user $u_i$ masks the trained model $x_t$ as follows: $c^{i, j}_t= \frac{x_t}{n}+ k_j- k_{j-1}\mod{q}$ and sends $<c^{i, j}_t, \textcolor{red}{\sigma^{i, j}_t}>$ to the $j^{th}$ node, where \textcolor{red}{$\sigma^{i, j}_t= \mathsf{SIG.Sign}(\mathsf{priv_{u_i}}, c^{i, j}_t)$}. 
                \end{itemize}
             \end{itemize}
            \begin{itemize}
                \item Intermediate Servers:
                \begin{itemize}
                     \item Each intermediate server, say $f_j$ chooses an empty list $\mathrm{F_j}$.
                    \item Each intermediate server $f_j$ receives the tuple $<c^{i, j}_t, \textcolor{red}{\sigma^{i, j}_t}>$ from each participating user $u_i$. 
                    \item \textcolor{red}{The intermediate server computes $\mathsf{SIG.Ver}(\mathsf{pub_{u_i}}, c^{i, j}_t, \sigma^{i, j}_t)$. If the signature is valid}, it adds the user identity $u_i$ to the participating user list $\mathrm{F_j}$. 
                    \item Each intermediate server $f_j$ sends the tuple $<\mathrm{F_j}, \textcolor{red}{\sigma_{\mathrm{F_i}}}>$ to the aggregator if and only if $|\mathrm{F_j}|\geq t$, \textcolor{red}{where $\sigma_{\mathrm{F_j}}= \mathsf{SIG.Sign}(\mathsf{priv_{f_j}}, \mathrm{F_j})$}. Otherwise, it aborts.
                \end{itemize}
            \end{itemize}
            \begin{itemize}
                \item Aggregator $\mathsf{Agg}$:
                \begin{itemize}
                    \item The aggregator chooses an empty list $\mathrm{A}$.
                   \item The aggregator receives the tuple $<c^{i, \mathsf{Agg}}_t, \textcolor{red}{\sigma^{i, \mathsf{Agg}}_t}>$ from each user $u_i$. 
                   \item \textcolor{red}{It performs $\mathsf{SIG.Ver}(\mathsf{pub_{u_i}}, c^{i, \mathsf{Agg}}_t, \sigma^{i, \mathsf{Agg}}_t)$. If the signature is valid}, it adds the user identity $u_i$ to the participating user list $\mathrm{A}$.
                   \item If $|\mathrm{A}|\geq t$ and the aggregator receives all tuples $<\{\mathrm{F_j}, \textcolor{red}{\sigma_{\mathrm{F_j}}}\}_{\forall j\in (\mathcal{N}\setminus \mathrm{Agg})}>$ from all the intermediate servers, \textcolor{red}{it verifies the signatures $\{\mathsf{SIG.Ver}(\mathsf{pub_{f_j}}, \mathrm{F_j}, \sigma_{\mathrm{F_j}})\}_{\forall j \in (\mathcal{N}\setminus \mathrm{Agg})}$}; otherwise it aborts. 
                   \item The aggregator computes the common active user list $\mathrm{I}= \mathrm{A}\cap \{\mathrm{F_j}\}_{\forall j \in (\mathcal{N}\setminus \mathrm{Agg})}$.
                   \item If $\mathrm{|I|}\geq t$, the aggregator sends back the tuple $<\mathrm{I}, \textcolor{red}{\sigma_{\mathsf{Agg}}}>$ to each intermediate server $f_j$, \textcolor{red}{where $\sigma_{\mathsf{Agg}}= \mathsf{SIG.Sign(\mathsf{priv_{Agg}}, I)}$}.
                \end{itemize}
            \end{itemize}
            \item \textbf{Partial Aggregation by the Intermediate Servers}
                \begin{itemize}
                   \item Each participating intermediate server, say $f_j$ receives the tuple $<\mathrm{I}, \textcolor{red}{\sigma_{\mathsf{Agg}}}>$ from the aggregator.
                   \item If $|\mathrm{I}|\geq t$, \textcolor{red}{the intermediate server $f_j$ checks $\mathsf{SIG.Ver}(\mathsf{pub_{Agg}}, \mathrm{I}, \sigma_{\mathsf{Agg}})$. If the signature is valid}, the intermediate server ${f_j}$ computes the partially aggregated masked model $\mathsf{PartialAgg_{f_j}}$ using the masked model received from the users in $\mathrm{I}$.
                   $$\mathsf{PartialAgg_{f_j}}= \sum_{\forall i\in \mathrm{I}}c^{i, j}_t$$ 
                   \item The intermediate server $f_j$ sends the tuple $<\mathsf{PartialAgg_{f_j}}, \textcolor{red}{\sigma_{\mathsf{PartialAgg_{f_j}}}}>$ to the aggregator, where \textcolor{red}{$\sigma_{\mathsf{PartialAgg_{f_j}}}= \mathsf{SIG.Sign}(\mathsf{priv_{f_j}}, \mathsf{PartialAgg_{f_j}})$}.
                \end{itemize}
            \item \textbf{Final Aggregation and Unmasking by the Aggregator} 
                \begin{itemize}
                   \item Once the aggregator receives all the tuples $\{\mathsf{PartialAgg_{f_j}}, \textcolor{red}{\sigma_{\mathsf{PartialAgg_{f_j}}}}\}_{\forall f_i \in (\mathcal{N}\setminus \mathrm{Agg})}$, \textcolor{red}{it checks $\{\mathsf{SIG.Ver}(\mathsf{pub_{f_j}}, \mathsf{PartialAgg_{f_j}}, \sigma_{\mathsf{PartialAgg_{f_j}}})\}_{\forall j \in (\mathcal{N}\setminus \mathrm{Agg})}$. If the signatures are valid}, the aggregator performs the final aggregation as follows:
                   $$\theta_t=\sum_{\forall i\in \mathrm{I}}c^{i, \mathrm{Agg}}_t+ \sum_{\forall j\in (\mathcal{N}\setminus \mathrm{Agg})}\mathsf{PartialAgg_{f_j}}$$
                   \item Aggregator sends the global model parameter $\theta_t$ and the common active user list $\mathrm{I}$ to each participating user for the next round of training.
                \end{itemize}
            \item \textbf{Global Model Parameter Verification}
                \begin{itemize}
                \item Aggregator:
                \begin{itemize}
                  \item Aggregator chooses a random vector $s_t\in \mathbb{Z}_q$. 
                  \item Aggregator computes $R= H(\theta_t)+ s_t$ and a message authentication code $S=\mathtt{MAC}_{s_t}(\theta_t)$
                  \item Aggregator sends the tuple $V= <\mathrm{T}= <R, S>, \mathrm{A}, \textcolor{red}{\mathsf{SIG.Sign(\mathsf{priv_{Agg}}, <T, \mathrm{I, A}>)}}>$ to each intermediate server.
                \end{itemize}
            \item Intermediate Servers:
                \begin{itemize}
                    \item \textcolor{red}{Each intermediate server, say $f_j$ verifies $\mathsf{SIG.Sign(\mathsf{priv_{Agg}}, <\mathrm{T}, \mathrm{I, A}>)}$ once it received the tuple $<\mathrm{T}= <R, S>, \mathrm{A}, \mathsf{SIG.Sign(\mathsf{priv_{Agg}}, <T, \mathrm{I, A}>)}>$ from the aggregator.}
                    \item \textcolor{red}{If the verification is successful,} the intermediate server $f_j$ forwards the tuple $<\mathrm{T}$, $\mathrm{I, A}>$ to the users in $\mathrm{I}$ along with the user list $\mathrm{F_j}$.
                \end{itemize}
                \item User:
                \begin{itemize}
                     \item Each user $u_i$ receives the tuple $<\mathrm{T}=<R, S>, \mathrm{I, A}, \textcolor{red}{\mathsf{SIG.Sign(\mathsf{priv_{Agg}}, <T, \mathrm{I, A}>)}}>$ and the user list $\{\mathrm{F_j}\}_{\forall j\in (\mathcal{N}\setminus \mathrm{Agg})}$ from the intermediate servers.
                    \item \textcolor{red}{The user first verifies $\mathsf{SIG.Sign(\mathsf{priv_{Agg}}, <T, \mathrm{I, A}>)}$. If verification is successful,} it verifies if $\mathrm{I}= \mathrm{A}\cap \{\mathrm{F_j}\}_{\forall j \in (\mathcal{N}\setminus \mathrm{Agg})}$ and $|\mathrm{I}|\geq t$. If verification fails, the user stops from participating in future rounds. Otherwise, the user goes to the following steps.
                    \item The user $u_i$ recovers the secret key $s_t= R- H(\theta_t)$ (where the user already has $\theta$ from the aggregator). 
                    \item If $S= \mathtt{MAC_{s_i}}(\theta_t)$, the user $u_i$ compares the remaining $S$ values received from the other intermediate servers to verify their consistency. If the verification fails or any mismatches are found, the user $u_i$ detects a model inconsistency attack and ceases participation.
                \end{itemize}
             \end{itemize}
        \end{itemize}
       
      \end{minipage}
    };  
  \end{tikzpicture}
  \caption{Detailed Description of Our Proposed Secure Aggregation Protocol}
  \label{fig:protocol-box}
  \end{scriptsize}
\end{figure*}

\section{Our Proposed Protocol}
\label{sec:our_scheme}
In this section, we provide a detailed explanation of our proposed protocol. We first present the technical intuition behind the protocol, followed by a description of its main construction. 

\subsection{Technical Intuition}
\label{sec:tech_intuition}
Our protocol has three main entities: users, intermediate servers, and the aggregator, as described in Section \ref{sec:system-architecture}. Users mask their locally trained models, intermediate servers partially aggregate these masked models, and the aggregator completes the final aggregation, eventually unmasking the global model. In this section, we explain how masking and unmasking work, along with how our protocol manages user dropouts and detects model inconsistency. Let $x_u$ be the private vectors of a user $u$, and $\mathcal{U}$ be the set of participating users in a round. 
\par 
The main goal of our protocol is to compute $\sum_{\forall u \in \mathcal{U}} x_u$ in such a way that the aggregator cannot see the individual private vectors $x_u$ of any user. Our protocol can tolerate up to $m-2$ compromised users, where $m = |\mathcal{U}|$. This means that the aggregator will only get at most the aggregated vectors of the two honest users, without learning their individual private vectors. Moreover, our protocol ensures that no user learns any useful information about another user's private vector.

\paragraph*{\textbf{Masking}} 
In our system, there are $m$ participating users, $n$ intermediate servers along with the aggregator. We represent the set of intermediate servers and the aggregator as $\mathcal{N}$. Each user $u$ generates a random vector $k_i$ in $\mathbb{Z}_q$ for each $i\in \mathcal{N}$. The user $u$ masks its model $x_u$ using the formula:
\begin{align}
    c^i=& x_u + k_i \mod{q}.
\end{align}
 This masked model $c^i$ is then sent to the $i^{th}$ entity in $\mathcal{N}$. All other users follow the same process. After receiving all masked models from the participating users, the intermediate servers partially aggregate the masked models and send the results to the aggregator. The aggregator then combines these partially aggregated masked models with its own to produce the final aggregated masked model. 
 \par 
 We can observe that these masked models do not reveal any information about the user $u$'s private vector $x_u$ to the intermediate servers, the aggregator, or other users since the random vectors $\{k_i\}_{\forall i\in \mathcal{N}}$ are private to the user $u$ and function as one-time pads. 

 \paragraph*{\textbf{Unmasking using Key Negation Technique}} The masking process described earlier does not involve an unmasking step, meaning the aggregator only receives the combined masked models. To introduce the unmasking step, our protocol uses a key negation approach, which is based on the aggregated homomorphic encryption method outlined in Equation \ref{eq:agg_cipher} (see Section \ref{sec:symmetricKeyHomoEncrypt}). The main goal is to use two random secret vectors to mask the models for each element in $\mathcal{N}$ so that these vectors cancel each other out during aggregation.
\par 
The concept behind our key negation mechanism is shown in Figure \ref{fig:key-neg}. Each node, labeled $c_i$, represents a masked model for the $i^{th}$ entity in $\mathcal{N}$, where $c_i = \mathsf{Enc_{(k_i- k_{i+1})}}(x_u) = x_u + \mathtt{k_i} - \mathtt{k_{i+1}} \mod{q}$. A directed arrow from node $c_i$ to node $c_{i+1}$ signifies the negation of key $k_{i+1}$ during the aggregation of masked models $c_i$ and $c_{i+1}$, as described in Equation \ref{eq:two-agg}.


\begin{flalign}
    \label{eq:two-agg}
    c_i+ c_{i+1}=& [x_u+ \mathtt{k_i}- \mathtt{k_{i+1}}] + [ x_u+ \mathtt{k_{i+1}}- \notag\\&\mathtt{k_{i+2}}]\mod{q}\notag\\
    =& [ x_u+ x_{u+1}+ \mathtt{k_i}- \cancel{\mathtt{k_{i+1}}}+ \cancel{\mathtt{k_{i+1}}}- \notag\\&\mathtt{k_{i+2}}] \mod{q}.
\end{flalign}
If a directed graph with a cycle, similar to the one shown in Figure \ref{fig:key-neg}, is generated from the entities in $\mathcal{N}$, aggregating all the masked models associated with the nodes will negate all the secret vectors, yielding the aggregated model. This process is illustrated in Equation \ref{eq:final-agg}.  
\begin{align}\label{eq:final-agg}
    \frac{\left(\sum_{1\leq i\leq n}c_i\right)}{n}=& n\cdot x_1+ n\cdot x_2+ \cdots+ n\cdot x_u+ (\cancel{k_1}-\cancel{k_2})+ (\cancel{k_2}-\cancel{k_1}+ \notag\\
    & \cdots+ (\cancel{k_{n-1}}- \cancel{k_n})+ (\cancel{k_n}-\cancel{k_1}))
    = \sum_{1\leq u\leq n}x_u.
\end{align}

\paragraph*{\textbf{User Dropouts}} In the BBGLR protocol and its variants, users agree on shared secret values with each other. This ensures that their random masks cancel out during the aggregation process. If a user drops out, the aggregator must still be able to remove that user's masked contribution. To do this, the aggregator collects the secret shares from the remaining users to cancel out the dropped user’s contribution in the final aggregation. A similar method is also used in the Flamingo protocol \cite{Ma2023}, where the aggregator must contact a set of \emph{decryptors} to handle the dropped users' contributions. 
\par 
In contrast, our protocol does not require users to agree on shared secrets or interact with each other. Each user’s masked model is independent of the others, unlike the BBGLR protocol. Our dropout mechanism is straightforward and requires only an additional communication round between the intermediate servers and the aggregator. After the masking phase, each user sends their masked models to both the intermediate servers and the aggregator. Consequently, each intermediate server $f_i$ and the aggregator generate a user list, $\mathrm{F_i}$ and $\mathrm{A}$ respectively, based on the users from whom they received masked models. Finally, the aggregator obtains the common active user list, $\mathrm{I}$ by collecting the user lists $\mathrm{F_i}$ from all intermediate servers. This list represents the active users for that round, from whom all entities received a masked model. All the intermediate servers and the aggregator use the user list $\mathsf{I}$ to aggregate the masked models. This process ensures that dropped users are identified, and their contributions are excluded from the aggregation without compromising the global model.
\paragraph*{\textbf{Model Inconsistency Attack Detection}}
In \cite{Pasquini2022}, it is explained how a malicious server (or the aggregator in our case) can carry out a model inconsistency attack. In this type of attack, as we explained earlier, the server manipulates the global model to provide different versions to specific users, potentially exposing sensitive information about them. In \cite{Pasquini2022}, one of the suggested ways to prevent issues is to make sure all users get the exact same global model from the aggregator. Our protocol aims to use this technique.
\par 
In our protocol, the aggregator sends the global model along with a MAC to the intermediate servers. These servers forward the MACs to all active users in $\mathsf{I}$. Since each user receives the same MACs from every intermediate server, they can easily detect inconsistencies by checking for mismatches. If a user observes any discrepancy, they stop participating in the training process, as this signals a model inconsistency attack\footnote{Our malicious model uses digital signatures to detect any tampering with MACs or other messages, preventing false accusations between participants.}.
\paragraph*{\textbf{Forward and Backward secrecy}}
One of the main limitations of e-SeaFL and Flamingo is their lack of FS and BS, due to reliance on a long-term secret seed as discussed in Appendix \ref{appendix:fs-bs-analysis}. This seed is used to generate per-element masking values for each FL round via PRFs. If the seed is compromised, all associated past and future masking values can be recomputed. Our protocol avoids this problem by using fresh, random masking vectors from $\mathbb{Z}_q$ in every training round. These masks are never reused and are not derived from any long-term secret. As a result, even if an adversary learns the masks from one round, it does not affect the security of previous or future rounds. Each round is thus cryptographically independent, providing strong FS and BS guarantees.

\paragraph*{\textbf{Dynamic User Settings}}
As described earlier, our protocol does not require users to pre-establish shared secret seeds to participate in a training round, unlike existing protocols such as e-SeaFL and Flamingo. Instead, users can independently select fresh, random masking vectors from $\mathbb{Z}_q$ in each training round. These vectors are eventually canceled out during aggregation, thanks to our key negation technique and the intermediate server layer. There is no requirement to register users during the initial setup phase for key exchange. Under our security model, any permitted new user can join and leave the FL training at any round without compromising correctness or security. 

\subsection{Main Protocol}
\label{sec:main-scheme}
Our protocol involves $m$ users and $n$ intermediate servers, along with an aggregator. Each user $u_i \in \mathcal{U}$ holds a private vector $x_t$ of size $n$ for the round $t$, where each element of $x_t$ belongs to $\mathbb{Z}_q$ for some $q$. Both $m$ and $n$ are polynomially bounded. The $i^{th}$ user and $j^{th}$ intermediate server are denoted as $u_i$ and $f_j$, respectively. The security of our protocol relies on the initial security parameter $\lambda$, which defines the overall strength of the protocol. Our protocol also uses a PKI, assigning each entity a private and public key pair for signing messages in the malicious model. The detailed description of our proposed secure aggregation protocol is shown in Figure \ref{fig:protocol-box}. 
\par 
The aggregator begins the protocol by sharing the initial global model with all participating users. The aggregator and the users communicate over private and authenticated channels. Each user then trains the model using their local datasets. Once trained, users mask their local models using the key negation mechanism described earlier and send the masked models to both the intermediate servers and the aggregator. This communication between users and intermediate servers also takes place over private and authenticated channels. 
\par 
The intermediate servers and aggregator collect masked models from at least $t$ users, where $t$ is the threshold required for each training round. This threshold $t$ is critical for maintaining the security of the masked models. For example, if only two users participate in a round, one user’s local model could be exposed to the other, increasing the risk of collusion. Therefore, the larger the value of $t$, the more secure the system becomes. The intermediate servers and the aggregator wait for a certain timeout period to gather enough masked models before aborting the process if an insufficient number is collected. Users are free to drop out and join at any time.
\par 
Once the intermediate servers receive the required number of masked models, they send their list of participating users, denoted as $\mathrm{F_i}$, to the aggregator if and only if $|\mathrm{F_i}|\geq t$. The aggregator generates a common active user list $\mathrm{I}$ from the lists received from the intermediate servers ${\mathrm{F_i}}_{\forall i\in (\mathcal{N}\setminus \mathrm{Agg})}$ and its own user list, $\mathrm{A}$. The aggregator then sends the list $\mathrm{I}$ to the intermediate servers. If $|\mathrm{I}|\geq t$, the intermediate servers partially aggregate the masked models of users in $\mathrm{I}$. Finally, the aggregator combines the partially aggregated values from the intermediate servers with the masked models from users in $\mathrm{I}$ to produce the plaintext global model ($\theta$) for the current training round.
\par 
Afterward, the aggregated global model $\theta$ is shared with all active participants for the next round of training. Simultaneously, the aggregator sends the tuple $V$ to each intermediate server. The intermediate server forwards a part of the tuple $V$ to the users in $\mathrm{I}$ for verification. If the verification of the tuples is successful, the users participate in the next round; otherwise, they abort.
\par 
Figure \ref{fig:protocol-box} shows two versions of the protocol: one for the semi-honest model and one for the malicious model. In the semi-honest model, where all entities adhere to the protocol without deviation, signatures and the PKI are not needed; these parts are marked in red to indicate they can be skipped. In the malicious model, signatures and the PKI are required to ensure proper authentication, message integrity, and non-repudiation.

\section{Security Analysis}
\label{sec:sec_analysis}
In this section, we present the security claims and their respective proofs for our proposed protocol discussed in Section \ref{sec:main-scheme}. 

\subsection{Semi-Honest Model}
\label{sec:Honest-but-Curious}
In this section, we demonstrate that even if a threshold number of users, denoted as $t_c$, and intermediate servers, denoted as $t_f$, collude among themselves or with the aggregator, they cannot learn about the remaining genuine users' local weighted models. We follow a similar security model to those in \cite{Bonawitz2017} and \cite{Liu2023}. 
\par 
We consider three scenarios: first, a subset of users is dishonest and collude with adversaries or among themselves, while the intermediate servers and the aggregator remain honest. Second, a subset of intermediate servers is dishonest and can collude, while the users and the aggregator are honest. Third, a subset of users and intermediate servers, along with the aggregator, are dishonest and can collude. Our goal is to show that if fewer than $t_c$ (where $t_c>2$) and $t_f$ (where $t_f> 1$) entities are compromised, our protocol still protects the individual locally trained models of the remaining genuine users. We assume that $\mathbb{U}$ and $\mathbb{F}$ are the total number of users and intermediate servers, respectively. We also assume that $\mathcal{U}_i$ are the number of participating users in the $i^{th}$ round, and $x_{\mathbb{U}}$ is the locally trained models of the users $\mathbb{U}$.
\par 
Let $\mathsf{Real_{C}^{\mathbb{U}, t_c, \lambda}}(x_{\mathbb{U}}, \mathbb{U}_1, \mathbb{U}_2, \mathbb{U}_3)$ be a random variable representing the joint views of participants in $C$ during the real execution of our protocol. Let $\mathsf{Sim_{C}^{\mathbb{U}, t_c, \lambda}}(x_{C}, \mathbb{U}_1, \mathbb{U}_2, \mathbb{U}_3)$ be the combined views of participants in $C$ when simulating the protocol, with the inputs of honest participants selected randomly and uniformly, denoted by $x_C$. Following the aforementioned idea, the distributions of $\mathsf{Real_{C}^{\mathbb{U}, t_c, \lambda}}(x_{\mathbb{U}}, \mathbb{U}_1, \mathbb{U}_2, \mathbb{U}_3)$ and $\mathsf{Sim_{C}^{\mathbb{U}, t_c, \lambda}}(x_{C}, \mathbb{U}_1, \mathbb{U}_2, \mathbb{U}_3)$ should be indistinguishable.


\begin{theorem}{(Security Against Semi-Honest Users only)}\label{theorem:1}
    For all $\mathbb{U}, \mathbb{F}, \mathsf{t_c}, \lambda$ with $|C_c|< \mathsf{t_c}, x_{\mathbb{U}}, \mathbb{U}_1, \mathbb{U}_2, \mathbb{U}_3$ and $C_c$ such that $C_c\subseteq \mathbb{U}, \mathbb{U}_3\subseteq \mathbb{U}_2 \subseteq\mathbb{U}_1 \subseteq \mathbb{U}$, there exists a probabilistic-time (PPT) simulator $\mathsf{Sim_{C}^{\mathbb{U}, t_c, \lambda}}$ which output is perfectly indistinguishable from the output of $\mathsf{Real_{C}^{\mathbb{U}, t_c, \lambda}}$:
    $$\mathsf{Sim_{C}^{\mathbb{U}, t_c, \lambda}}(x_{C_c}, \mathbb{U}_1, \mathbb{U}_2, \mathbb{U}_3)\equiv \mathsf{Real_{C}^{\mathbb{U}, t_c, \lambda}}(x_\mathbb{U}, \mathbb{U}_1, \mathbb{U}_2, \mathbb{U}_3)$$
\end{theorem}
\begin{proof}\label{proof:1}
We present the detailed security proof in the Appendix \ref{appendix:Semi-Honest-Users-Only}.
\end{proof}

\begin{theorem}{(Security Against Semi-Honest Intermediate Servers only)}\label{theorem:2}
    For all $\mathbb{U}, \mathbb{F}, \mathsf{t_f}, \lambda$ with $|C_f|\leq \mathsf{t_f}, x_{\mathbb{U}}, \mathbb{U}_1, \mathbb{U}_2, \mathbb{U}_3$ such that $C_f\subseteq \mathbb{F}$, $\mathbb{U}_3\subseteq \mathbb{U}_2 \subseteq\mathbb{U}_1 \subseteq \mathbb{U}$, there exists a probabilistic-time (PPT) simulator $\mathsf{Sim_{C}^{\mathbb{F}, t_f, \lambda}}$ which output is perfectly indistinguishable from the output of $\mathsf{Real_{C}^{\mathbb{F}, t_f, \lambda}}$:
    $$\mathsf{Sim_{C}^{\mathbb{F}, t_f, \lambda}}(x_{C_f}, \mathbb{U}_1, \mathbb{U}_2, \mathbb{U}_3)\equiv \mathsf{Real_{C}^{\mathbb{F}, t_f, \lambda}}(x_\mathbb{U}, \mathbb{U}_1, \mathbb{U}_2, \mathbb{U}_3)$$
\end{theorem}
\begin{proof}\label{proof:2}
We present the detailed security proof in the Appendix \ref{appendix:semi-honest-intermediate-servers-only}.
\end{proof}
\begin{theorem}{(Security Against Semi-Honest Aggregator and Semi-Honest Intermediate Servers)}\label{theorem:3}
    For all $\mathbb{U}, \mathbb{F}, \mathsf{t_c, t_f}, \lambda$ with $|C_c|\leq \mathsf{t_c}, |C_f|\leq \mathsf{t_f}, x_{\mathbb{U}}, \mathbb{U}_1, \mathbb{U}_2, \mathbb{U}_3$ such that $C_c\subseteq \mathbb{U}, C_f\subseteq \mathbb{F}$, $\mathbb{U}_3\subseteq \mathbb{U}_2 \subseteq\mathbb{U}_1 \subseteq \mathbb{U}$, there exists a probabilistic-time (PPT) simulator $\mathsf{Sim_{C}^{\mathbb{U}, \mathbb{F}, t_c, t_f, \lambda}}$ which output is perfectly indistinguishable from the output of $\mathsf{Real_{C}^{\mathbb{U}, \mathbb{F}, t_c, t_f, \lambda}}$:
      \begin{align}
    &\mathsf{Sim_{C}^{\mathbb{U}, \mathbb{F}, t_c, t_f, \lambda}}(x_{C_c}\cup x_{C_f}, \mathbb{U}_1, \mathbb{U}_2, \mathbb{U}_3)\equiv \notag\\
    &\mathsf{Real_{C}^{\mathbb{U}, \mathbb{F}, t_c, t_f, \lambda}}(x_\mathbb{U}, \mathbb{U}_1, \mathbb{U}_2, \mathbb{U}_3) \notag
    \end{align}
\end{theorem}
\begin{proof}
    We present the detailed security proof in the Appendix \ref{proof:semihonest}. 
\end{proof}

\subsection{Malicious Model}
\label{sec:active-adversary}
In this section, we follow the standard security model as in \cite{Bonawitz2017} and \cite{Liu2023} to demonstrate that our protocol is secure against active malicious attackers. 
\par 
In this security model, we need to consider three cases. The first case is the \emph{Sybil Attack}, where the adversaries can simulate a specific honest user $u_i$ to get its inputs. Please note that our protocol is resistant to this attack, as it uses digital signatures to verify the origin of the messages. The second case involves the aggregator sending different lists of participating users in a round to the honest intermediate servers. This could lead to the disclosure of the local weighted models of the honest users. Our protocol detects this type of attack during the model parameter verification phase, where each user can verify the list of participating users in a round received from the intermediate servers. Any difference in the lists will indicate an attack. 
\par 
The third case involves an aggregator dishonestly dropping honest users at any round. We use the Random Oracle model to prove that our protocol is secure against any malicious dropout of honest users by adversaries. Let's consider $M_c$ as a probabilistic polynomial-time algorithm representing the \emph{next message} function of participants in $C$, which enables users in $C$ to dynamically select their inputs at any round of the protocol and the list of participating users. The Random Oracle can output the sum of a dynamically selected subset of honest clients for the simulator $\mathrm{Sim}$, which is indistinguishable from the combined view of adversaries in the real protocol execution $\mathsf{Real}(M_c)$. There are three possible scenarios for this type of active attack: first, the malicious users can collude among themselves; second, the malicious intermediate servers can collude among themselves; and finally, the malicious users and intermediate servers can collude with the aggregator. We present the following three theorems to demonstrate our security proofs.

\begin{table*}[htbp]
\centering
\caption{Running Time (ms) for User in Different Phases with a Vector Size of $48k$ and $64$-bit Length.}
\label{table:running-time-semi-honest-model}
\begin{tabular}{|p{3.5cm}|cc|cc|cc|cc|cc|}
\hline
Number of Users & \multicolumn{2}{c|}{$100$} & \multicolumn{2}{c|}{$200$} & \multicolumn{2}{c|}{$300$} & \multicolumn{2}{c|}{$400$} & \multicolumn{2}{c|}{$500$}\\ \hline
Scheme & \multicolumn{1}{c|}{e-SeaFL} & Ours & \multicolumn{1}{c|}{e-SeaFL} & Ours  & \multicolumn{1}{c|}{e-SeaFL} & Ours  & \multicolumn{1}{c|}{e-SeaFL} & Ours & \multicolumn{1}{c|}{e-SeaFL} & Ours\\ \hline
Masking Time$^*$ & \multicolumn{1}{c|}{$201432$} & $2193$ & \multicolumn{1}{c|}{$202927$} & $2324$ & \multicolumn{1}{c|}{$210059$} & $2487$ & \multicolumn{1}{c|}{$207008$} & $2329$ & \multicolumn{1}{c|}{$190656$} & $2199$ \\ \hline
Masking Time$^\#$ & \multicolumn{1}{c|}{202355} & $4896$  & \multicolumn{1}{c|}{209494} & $5113$ & \multicolumn{1}{c|}{212300} & $5241$ & \multicolumn{1}{c|}{205440} & $5004$ & \multicolumn{1}{c|}{207617} & $5218$ \\ \hline
Model Verification$^*$ & \multicolumn{1}{c|}{$93657498$} & $3864$ & \multicolumn{1}{c|}{92982871} & $3831$& \multicolumn{1}{c|}{92806374} & $3822$ & \multicolumn{1}{c|}{91344514} & $3845$ & \multicolumn{1}{c|}{91844664} & $3808$ \\ \hline
Model Verification$^\#$ & \multicolumn{1}{c|}{96533813} & $3734$ & \multicolumn{1}{c|}{96553259} & $3767$ & \multicolumn{1}{c|}{93273166} & $3826$ & \multicolumn{1}{c|}{96467621 } & $3962$ & \multicolumn{1}{c|}{95497198} & $3958$ \\ \hline
Total User Running Time$^*$ & \multicolumn{1}{c|}{194063757} & $6057$ & \multicolumn{1}{c|}{185477547} & $6155$ & \multicolumn{1}{c|}{190296465} & $6309$ & \multicolumn{1}{c|}{188311203} & $6174$ & \multicolumn{1}{c|}{188888499} & $6007$ \\ \hline
Total User Running Time$^\#$ & \multicolumn{1}{c|}{195167142} & $8630$ & \multicolumn{1}{c|}{ 196882805 } &$8880$& \multicolumn{1}{c|}{192445184 } & $9067$ & \multicolumn{1}{c|}{194149205 } & $8966$ & \multicolumn{1}{c|}{195775738 } &  $9176$\\ \hline
\end{tabular}
\\
$*$ represent cost in the semi-honest model; \# represent cost in the malicious model.
\end{table*}

\begin{figure*}[htbp]
    \centering
    \begin{subfigure}{0.4\textwidth}
        \centering
        \includegraphics[width=6.5cm, height=4.5cm]{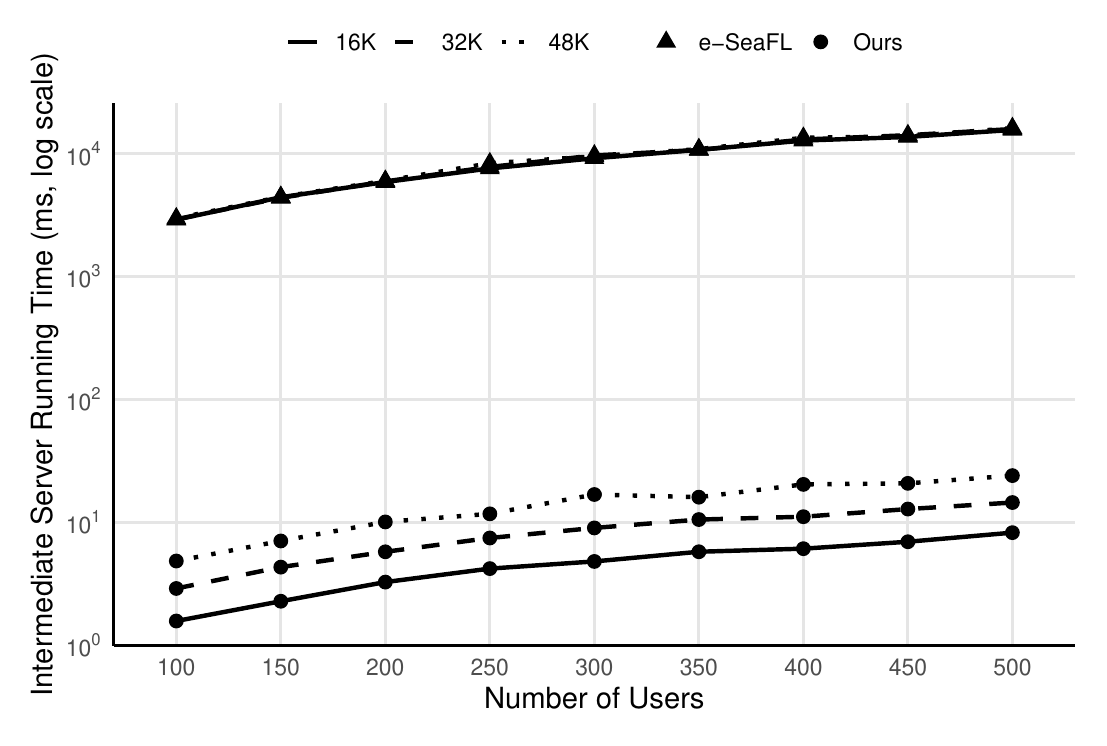}
        \caption{Intermediate Server}
        \label{fig:Semi-honest-Intermediate-Server-Time-with-Users}
    \end{subfigure}
    ~
    \begin{subfigure}{0.4\textwidth}
        \centering
        \includegraphics[width=6.5cm, height=4.5cm]{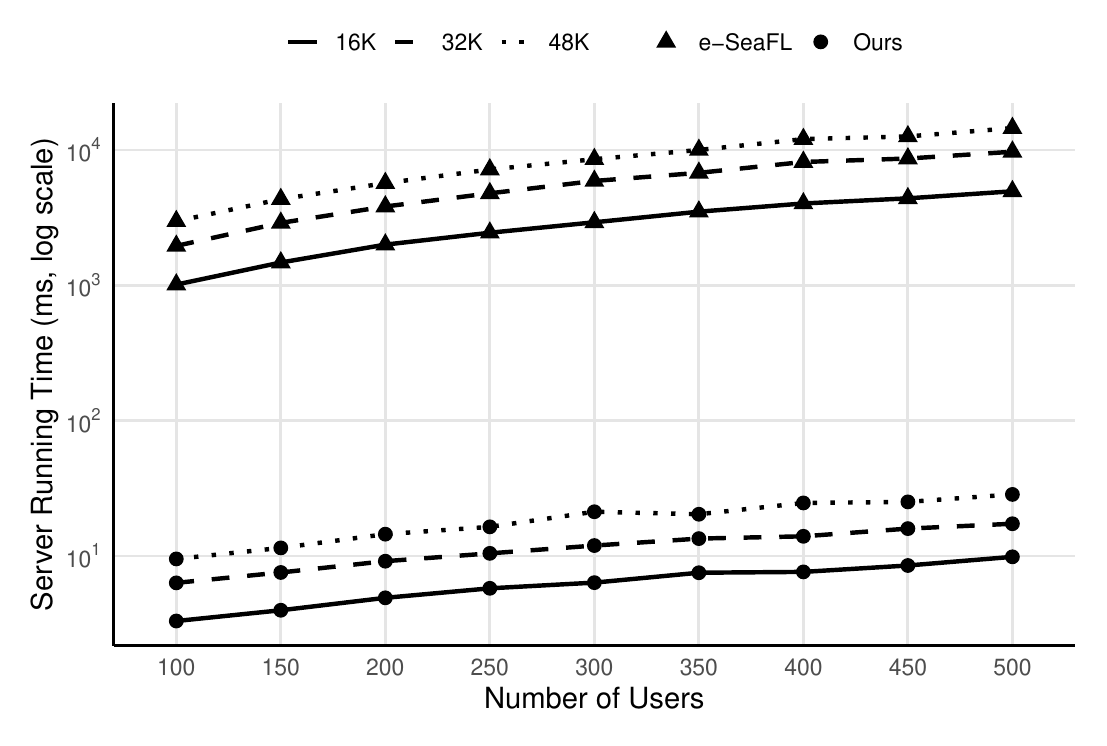}
        \caption{Aggregator}
        \label{fig:Semi-honest-Server-Time-With-Users}
    \end{subfigure}
    \caption{Average Running Time of Intermediate Server and Aggregator in the Semi-Honest Model per Round}
    \label{fig:1-semi-honest-ComputationTime-withoutDropout}
\end{figure*}

\begin{figure*}[htbp]
     \centering
     \begin{subfigure}[htbp]{0.44\textwidth}
         \centering
         \includegraphics[width=6.5cm, height=4.5cm]{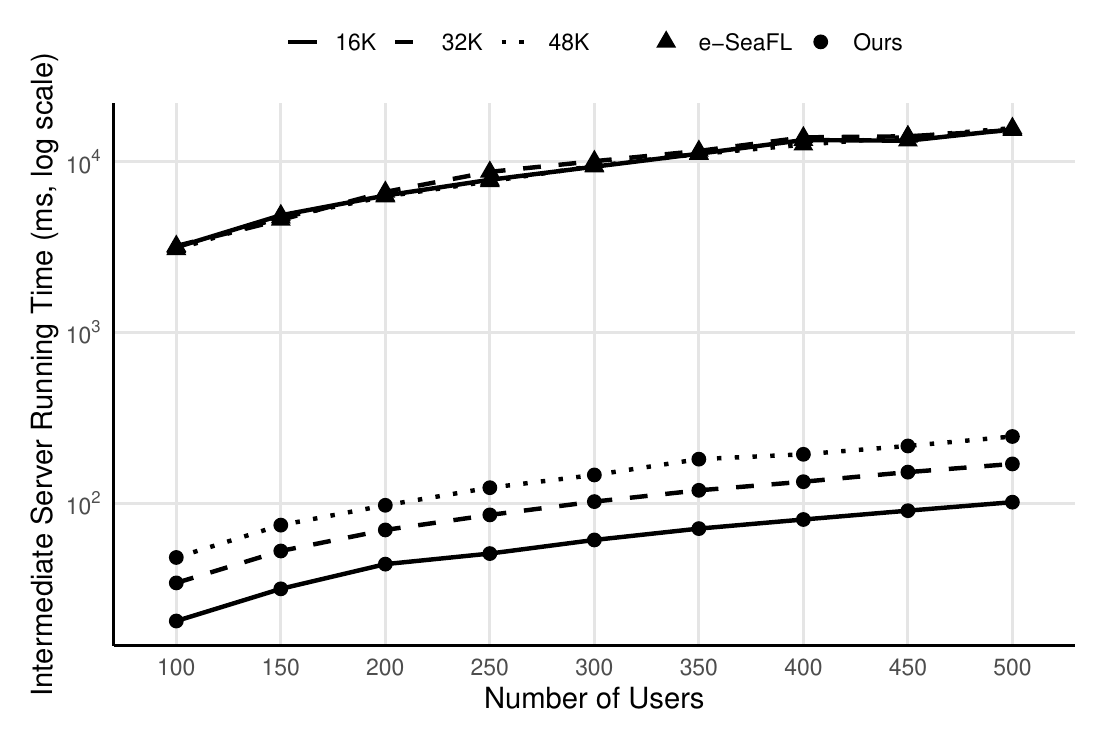}
         \caption{Intermediate Server}
         \label{fig:Malicious-Intermediate-Server-Time-with-Users}
     \end{subfigure}
     ~
     \begin{subfigure}[htbp]{0.44\textwidth}
         \centering
         \includegraphics[width=6.5cm, height=4.5cm]{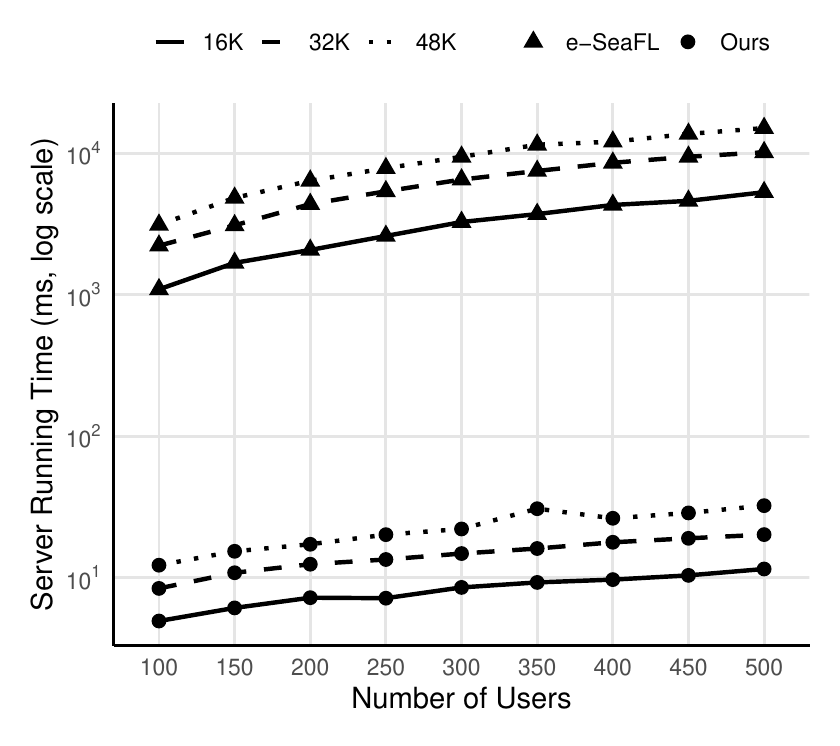}
         \caption{Aggregator}
         \label{fig:Malicious-Server-Time-with-Users}
     \end{subfigure}
        \caption{Average Running Time of Intermediate Server and Aggregator in the Malicious Model per Round}
        \label{fig:1-malicious-ComputationTime-withoutDropout}
\end{figure*}

\begin{figure*}[htbp]
     \centering
     \begin{subfigure}[htbp]{0.44\textwidth}
         \centering
         \includegraphics[width=6.5cm, height=4.5cm]{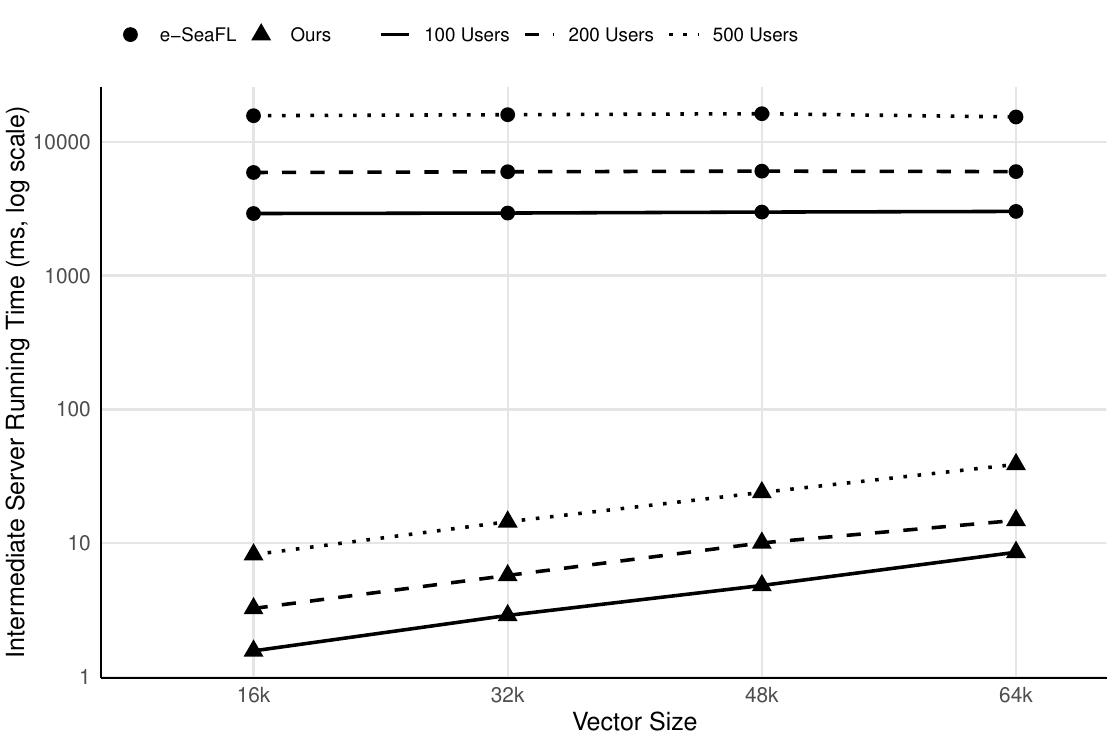}
         \caption{Intermediate Server}
         \label{fig:Semi-honest-Intermediate-Server-Time-With-Vector-Size}
     \end{subfigure}
     ~
     \begin{subfigure}[htbp]{0.44\textwidth}
         \centering
         \includegraphics[width=6.5cm, height=4.5cm]{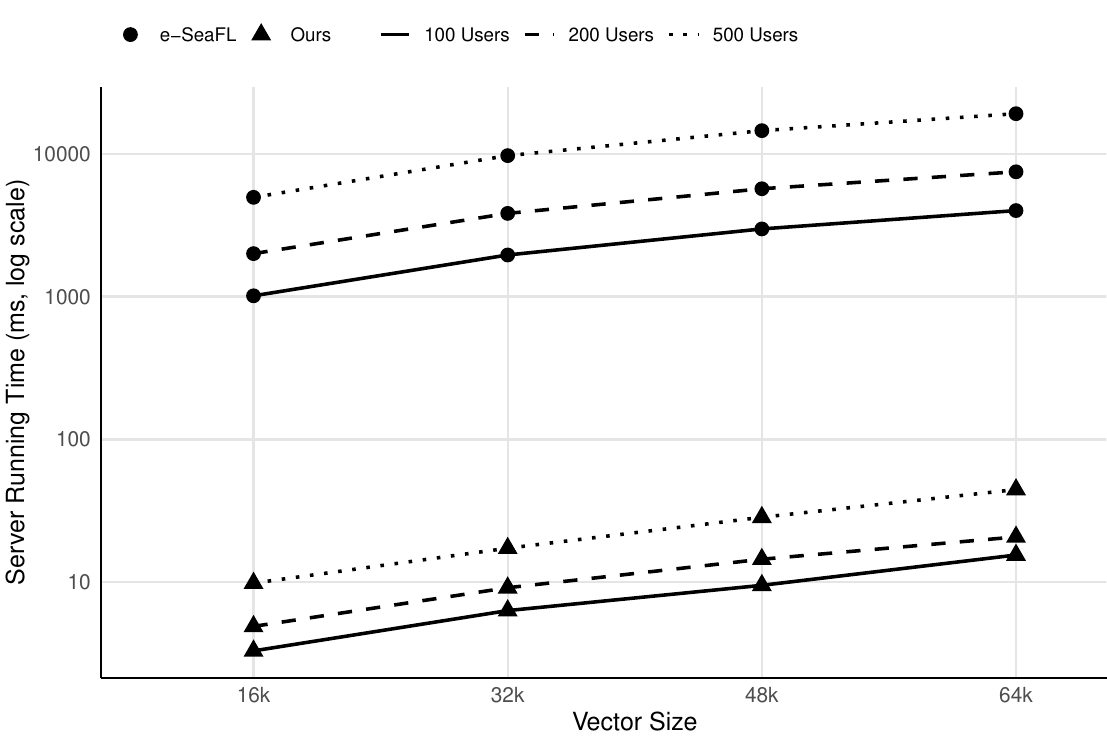}
         \caption{Aggregator}
         \label{fig:Semi-honest-Server-Time-With-Vector-Size}
     \end{subfigure}
        \caption{Average Running Time of Intermediate Server and Aggregator in the Semi-Honest Model per Round}
        \label{fig:2-semi-honest-ComputationTime-withoutDropout}
\end{figure*}

\begin{figure*}[htbp]
     \centering
     \begin{subfigure}[htbp]{0.4\textwidth}
         \centering
         \includegraphics[width=6.5cm, height=4.5cm]{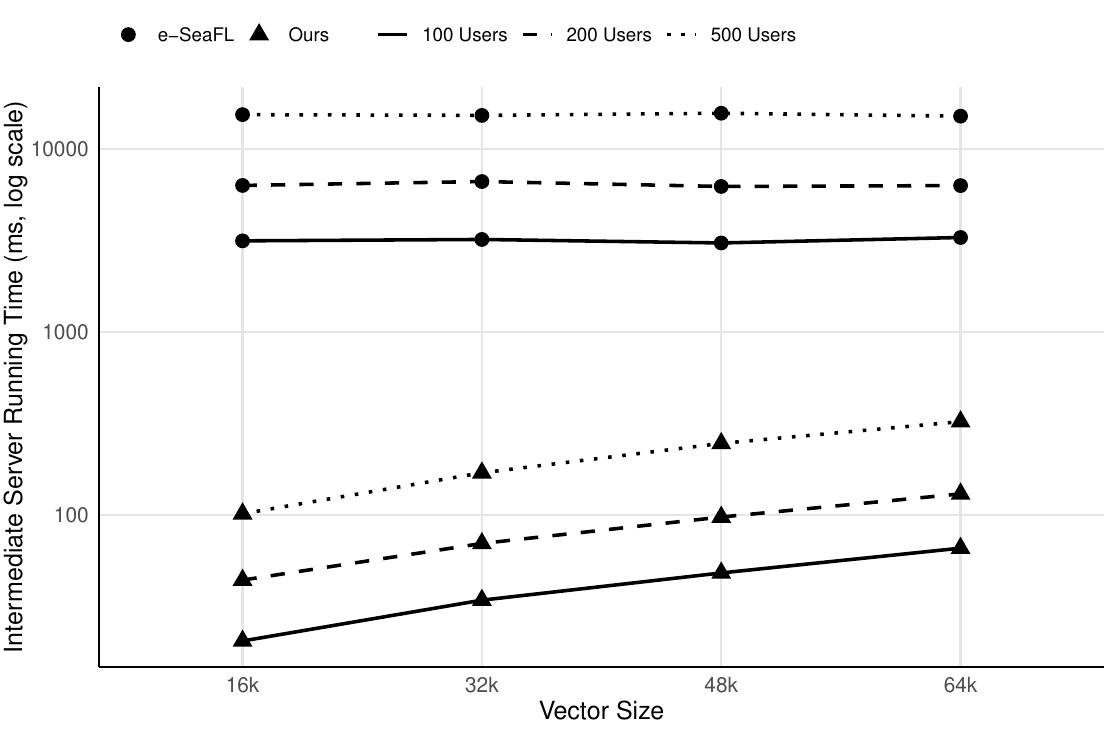}
         \caption{Intermediate Server}
         \label{fig:Malicious-Intermediate-Server-Running-Time-with-Vector-Siz}
     \end{subfigure}
     ~
     \begin{subfigure}[htbp]{0.4\textwidth}
         \centering
         \includegraphics[width=6.5cm, height=4.5cm]{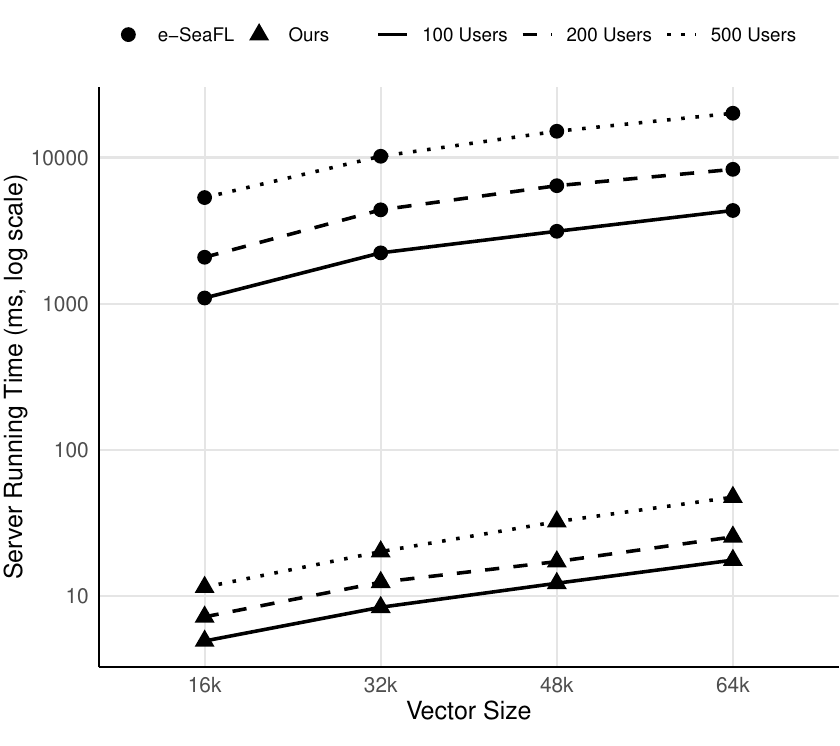}
         \caption{Aggregator}
         \label{fig:Malicious-Server-Running-Time-with-Vector-Size}
     \end{subfigure}
        \caption{Average Running Time of Intermediate Server and Aggregator in the Malicious Model per Round}
        \label{fig:2-malicious-ComputationTime-withoutDropout}
\end{figure*}

\begin{theorem}{(Security Against Dishonest Users only)}\label{theorem:4}
For all $\mathbb{U}, \mathbb{F}, \mathsf{t_c}, \lambda$ with $|C_c|\leq \mathsf{t_c}, x_{\mathbb{U}\setminus C_c}, \mathbb{U}$ and $C_c$ such that $C_c\subseteq \mathbb{U}$, there exists a probabilistic-time (PPT) simulator $\mathsf{Sim_{C}^{\mathbb{U}, t_c, \lambda}}$ which output is perfectly indistinguishable from the output of $\mathsf{Real_{C}^{\mathbb{U}, t_c, \lambda}}$:
$$\mathsf{Sim_{C}^{\mathbb{U}, t_c, \lambda}}(M_c, x_{\mathbb{U}\setminus C_c})\equiv \mathsf{Real_{C}^{\mathbb{U}, t_c, \lambda}}(M_c)$$
\end{theorem}
\begin{proof}
    This proof is similar to that of Theorem \ref{theorem:1}, as the users in $C_c$ gain no knowledge about $x_{\mathbb{U}\setminus C_c}$. Consequently, the simulator $\mathsf{Sim}$ can assign real inputs to the dishonest users while assigning dummy inputs to the remaining users, accurately replicating the perspectives of the users in $C_c$. Thus, the joint views of the users in $C_c$ in the simulation are indistinguishable from those in $\mathsf{Real}$.
\end{proof}

\begin{theorem}{(Security Against Dishonest Intermediate Servers only)}\label{theorem:5}
For all $\mathbb{U}, \mathbb{F}, \mathsf{t_f}, \lambda$ with $|C_f|\leq \mathsf{t_f}, x_{\mathbb{U}\setminus C_f}, \mathbb{U}$ and $C_f$ such that $C_f\subseteq \mathbb{F}$, there exists a probabilistic-time (PPT) simulator $\mathsf{Sim_{C}^{\mathbb{F}, t_f, \lambda}}$ which output is perfectly indistinguishable from the output of $\mathsf{Real_{C}^{\mathbb{F}, t_f, \lambda}}$:
$$\mathsf{Sim_{C}^{\mathbb{F}, t_f, \lambda}}(M_c, x_{\mathbb{F}\setminus C_f})\equiv \mathsf{Real_{C}^{\mathbb{F}, t_f, \lambda}}(M_c)$$
\end{theorem}
\begin{proof}
        We present the detailed security proof in the Appendix \ref{appendix:dishonest-intermediate-servers-only}.
\end{proof}

\begin{theorem}{(Security Against Dishonest Users, Intermediate Servers and Aggregator)}\label{theorem:6}
 For all $\mathbb{U}, \mathbb{F}, \mathsf{t_c, t_f}, \lambda$ with $|C_c|\leq \mathsf{t_c}, |C_f|\leq \mathsf{t_f}, x_{\mathbb{U}\setminus (C_c \bigcup C_f)}, \mathbb{U}$ such that $C_c\subseteq \mathbb{U}, C_f\subseteq \mathbb{F}$, $\delta_c= t_c- |C_c \cap \mathbb{U}|, \delta_f= t_f- |C_f \cap \mathbb{F}|$, there exists a probabilistic-time (PPT) simulator $\mathsf{Sim_{C}^{\mathbb{U}, t_c, t_f, \lambda}}$ which output is perfectly indistinguishable from the output of $\mathsf{Real_{C}^{\mathbb{U}, t_c, t_f, \lambda}}$:
    \begin{align*}
        &\mathsf{Sim_{C}^{\mathbb{U}, \mathbb{F}, t_c, t_f, \lambda}}(M_c, \{x_{\mathbb{U}\setminus (C_c\bigcup C_f)}\bigcup x_{\mathbb{F}\setminus C_f}\})\equiv \\
        &\mathsf{Real_{C}^{\mathbb{U}, \mathbb{F}, t_c, t_f, \lambda}}(M_c)
    \end{align*}
    where $\delta_c, \delta_f$ are the lower bound of the number of participating honest users and intermediate servers respectively.
\end{theorem}
\begin{proof}
    We present the detailed security proof in the Appendix \ref{proof:malicious}. 
\end{proof}

\section{Experimental Evaluation}
\label{sec:performance_eval}
In this section, we present experimental results to demonstrate the practical efficiency of our proposed protocol. A theoretical comparison of functionality and security is provided in Table~\ref{table:functionality-comparison}, while the complexity analysis of communication and computation costs is shown in Table~\ref{table:fun-com-commu-comparison}. These aspects were briefly discussed in Section~\ref{sec:related-work}. Overall, the theoretical findings indicate that our protocol offers enhanced functionality and stronger security guarantees, while maintaining communication and computation costs comparable to e-SeaFL, which is the most efficient single-setup-based secure aggregation protocol to date.
\par 
To further assess practicality, we conduct experiments comparing our protocol with e-SeaFL, focusing on concrete running times for the user, intermediate server, and aggregator. This evaluation highlights the computational efficiency of our protocol. Additionally, we report the accuracy achieved on the MNIST \cite{Lecun1998} and CIFAR-10 \cite{krizhevsky2009learning} datasets, comparing our results with those obtained using the standard FedAvg method \cite{mcmahan2017communication} (a baseline FL approach). 


We developed a prototype of our protocol and tested it on a VM with 2 vCPUs (Intel Xeon), 16GB of RAM, and Ubuntu Server 22.04. The prototype is implemented in Python (Python 3.7.11) using PyTorch 1.13.1, NumPy 1.21.6, and PyCryptodome 3.20. Our prototype employs the LeNet-5 architecture as an example to demonstrate both the learning performance and practical applicability of our protocol.
\par 
For performance testing, we omitted the actual model training process and instead generated a random data vector of the required size. This allowed us to obtain performance results comparable to those of e-SeaFL and to observe the scalability trend as the weight size increased. The implementation used for performance testing of our protocol is available at \url{https://anonymous.4open.science/r/scatesfl-815A/}. The real-world federated model training (using MNIST or CIFAR-10) based on our protocol can be found at \url{https://anonymous.4open.science/r/scatesfl-44E1/}. The corresponding e-SeaFL implementation is accessible at \url{https://anonymous.4open.science/r/scatesfl-6F59/}.
\begin{figure*}[htbp]
     \centering
     \begin{subfigure}[htbp]{0.4\textwidth}
         \centering
         \includegraphics[width=6.5cm, height=4.5cm]{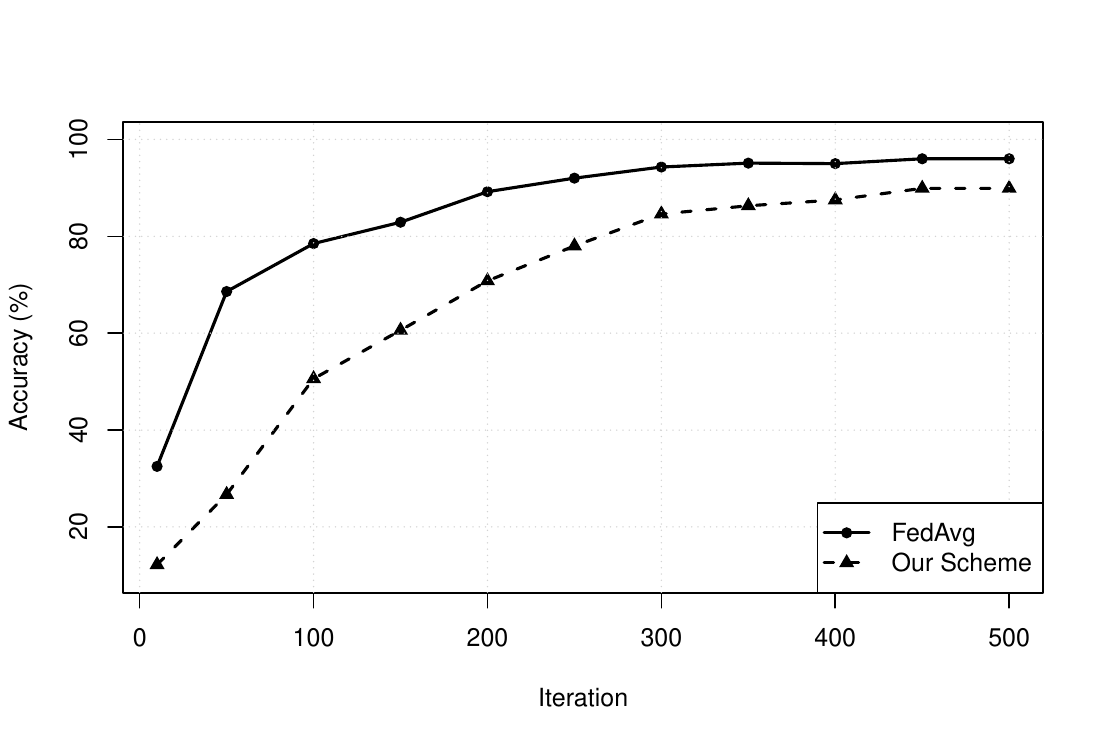}
         \caption{MNIST}
         \label{fig:MNIST-accuracy}
     \end{subfigure}
     ~
     \begin{subfigure}[htbp]{0.4\textwidth}
         \centering
         \includegraphics[width=6.5cm, height=4.5cm]{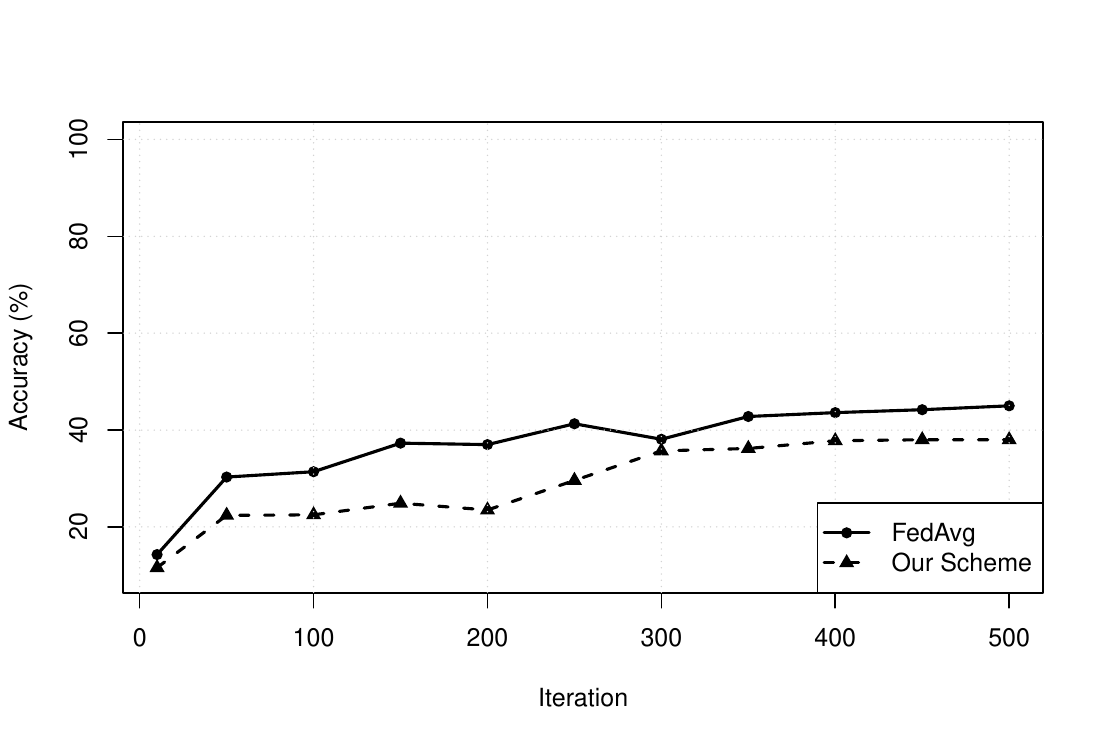}
         \caption{CIFAR-10}
         \label{fig:CIFAR-accuracy}
     \end{subfigure}
     \vspace{-.3cm}
        \caption{Accuracy Comparison Between Our Protocol and Baseline FedAvg \cite{mcmahan2017communication}}
        \vspace{-.2cm}
        \label{fig:accuracy}
\end{figure*}
\vspace{-.3cm}
\subsection{User Running Time Computation}
Table~\ref{table:running-time-semi-honest-model} presents the average per-user computation time across different phases- masking, model verification, and total running time- using a vector size of $48k$ with 64-bit elements and 5 intermediate servers per round. These values are averaged over five rounds and remain nearly constant across varying numbers of users, as the computation cost is primarily influenced by the input vector size. 
\par 
From Table~\ref{table:running-time-semi-honest-model}, it is evident that our protocol reduces masking time by about $98\%$ in the semi-honest model and $97\%$ in the malicious model, and reduces model verification time by about $99\%$ in both models, compared to e-SeaFL. As discussed earlier, e-SeaFL incurs higher computation due to additional PRF operations in the masking phase and elliptic curve cryptographic operations during model verification. Moreover, the reported masking time for e-SeaFL in our comparison excludes the computations related to the commitment phase (please refer e-SeaFL \cite{Behnia2024} for more details on the commitment phase); when these are included, the masking time increases to approximately $1.6$ minutes. Even without considering the commitment phase, our protocol demonstrates clear computational advantages, and when it is included, the performance gap becomes even more pronounced--- underscoring our protocol’s efficiency in both phases.
\par 
For the average total per-user running time (including the commitment cost for e-SeaFL), e-SeaFL requires approximately 194 million~$\mu$s per user, whereas our protocol requires only 6,140~$\mu$s, yielding an impressive $99\%$ reduction in the semi-honest model. In the malicious model, e-SeaFL averages 195 million~$\mu$s, compared to 8,944~$\mu$s for our protocol, resulting in a $99\%$ reduction. Overall, these results highlight the practicality and scalability of our protocol, especially for resource-constrained devices such as IoT nodes in federated learning environments.

\subsection{Intermediate Server and Aggregator Running Time Computation}
Figures~\ref{fig:Semi-honest-Intermediate-Server-Time-with-Users} and~\ref{fig:Semi-honest-Server-Time-With-Users} illustrate the computation times of the intermediate server and aggregator, respectively, under the semi-honest model, comparing our protocol with e-SeaFL across vector sizes of 16K, 32K, and 48K, and user counts up to 500. Experimental results show that our protocol reduces intermediate server computation time by approximately 600 to 2100 times, and aggregator computation time by approximately 300 to 580 times, compared to e-SeaFL. Figures~\ref{fig:Malicious-Intermediate-Server-Time-with-Users} and~\ref{fig:Malicious-Server-Time-with-Users} present results under the malicious model, where additional signature operations are required. Our experiments show that, compared to e-SeaFL, our protocol reduces intermediate server computation time by approximately $60$ to $165$ times, and aggregator computation time by approximately $220$ to $510$ times.
\par 
We further evaluated performance by varying the vector size (16K to 64K) for 100, 200, and 500 users. Figures~\ref{fig:Semi-honest-Server-Time-With-Vector-Size} and~\ref{fig:Malicious-Server-Running-Time-with-Vector-Size} show the average computation time of the aggregator in the semi-honest and malicious models, respectively. Similarly, Figures~\ref{fig:Semi-honest-Intermediate-Server-Time-With-Vector-Size} and~\ref{fig:Malicious-Intermediate-Server-Running-Time-with-Vector-Siz} show the average computation time of the intermediate server in the semi-honest and malicious models, respectively. From these figures, it is evident that our protocol reduces aggregator computation time by approximately $250$ to $560$ times in the semi-honest model and $220$ to $510$ times in the malicious model compared to e-SeaFL. Similarly, it reduces intermediate server computation time by approximately $350$ to $1900$ times in the semi-honest model and $47$ to $154$ times in the malicious model. 

\subsubsection{Accuracy Comparison}
We conducted experiments on two benchmark datasets, CIFAR-10 \cite{krizhevsky2009learning} and MNIST \cite{Lecun1998}, to assess the learning performance. CIFAR-10 has 60,000 color images of size $32 \times 32$ pixels, divided into 10 different classes. It has 50,000 training images and 10,000 test images. This dataset is used to test models on more challenging image tasks. MNIST has 70,000 grayscale images of handwritten digits, each $28 \times 28$ pixels. It includes 60,000 training images and 10,000 test images. MNIST is commonly used to check how well models learn simple image recognition tasks quickly.
\par 
Figures~\ref{fig:MNIST-accuracy} and \ref{fig:CIFAR-accuracy} present the accuracy of our protocol during training compared to the baseline FedAvg algorithm \cite{mcmahan2017communication}, over 500 rounds with 100 users. While FedAvg achieves slightly higher final accuracy, our protocol demonstrates steady and reliable learning throughout the process. On CIFAR-10, our protocol reaches 38\% accuracy versus FedAvg’s 45\%, and on MNIST it attains 89.9\%, which is close to FedAvg’s 96\%. The small drop in accuracy comes from encoding floating-point values as integers, where a minor loss of precision accumulates across rounds. This trade-off can be mitigated through techniques such as fixed-point encoding\cite{secureml,crypten} or quantization-aware training (QAT) \cite{jacob2018quantization,gholami2021survey}, which preserve numerical precision while remaining compatible with secure aggregation. Despite this, our protocol maintains strong accuracy while providing additional benefits in privacy, communication efficiency, and reliability, making it a practical and effective choice in settings with system constraints or strict security requirements.
\vspace{-.2cm}
\section{Conclusion}
\label{sec:conclusion}
In this paper, we presented a novel secure aggregation protocol tailored for dynamic FL environments. Our design addresses key limitations of existing protocols by introducing a masking protocol based on additive symmetric homomorphic encryption with key negation. This approach eliminates the need for repeated cryptographic setup, significantly reducing communication and computation overhead. Our protocol is the first single-setup protocol to simultaneously achieve forward and backward secrecy, model integrity verification, and dynamic user participation--- all while maintaining strong privacy guarantees and resilience to user dropouts. By incorporating intermediate servers, our protocol avoids user-to-user communication, reduces trust assumptions, and enhances scalability for real-world deployments. We also provided formal security proofs under both semi-honest and malicious adversarial models, demonstrating robustness. Our experimental evaluation shows that our protocol achieves up to $99\%$ reduction in user-side computation time compared to state-of-the-art protocols like e-SeaFL, while maintaining competitive model accuracy on benchmark datasets. Overall, our protocol offers a practical, efficient, and secure solution for FL in dynamic and resource-constrained environments, paving the way for broader adoption of privacy-preserving ML.
\par 
In our protocol, we assume that the intermediate servers remain online throughout the training process. In certain deployment scenarios, such as ad-hoc or intermittently connected edge networks, this assumption might not hold. Designing a protocol that supports dynamic intermediate servers while maintaining the same functionality, security, and efficiency is an important direction for future work. Additionally, further evaluation on real-world IoT testbeds is planned to validate the protocol’s performance and scalability under practical conditions, with a particular focus on communication overhead.

%
\section*{Acknowledgment}
The work has been supported by the Cyber Security Research Centre Limited whose activities are partially funded by the Australian
Government’s Cooperative Research Centres Program. 

\bibliographystyle{ACM-Reference-Format}

\appendix

\appendix
\section{Detailed Analysis of Forward and Backward Secrecy in e-SeaFL and Flamingo}
\label{appendix:fs-bs-analysis}
Both forward secrecy (FS) and backward secrecy (BS) are critical privacy properties in secure aggregation protocols, ensuring that the compromise of long-term cryptographic material does not expose users’ model updates from prior training rounds and any future training rounds, respectively. This protection is particularly important in FL, where training may span multiple rounds over days, and client devices are often deployed in untrusted or adversarial environments. Without FS and BS, a one-time key compromise can retroactively invalidate the privacy of all prior and future contributions, respectively.
\par 
As discussed earlier, single-setup secure aggregation protocols offer a promising direction for scalable, efficient, and secure FL. Compared to multi-setup protocols, they significantly reduce communication and computation overhead by avoiding repeated cryptographic setup in each training round. Two notable protocols in this category are e-SeaFL \cite{Behnia2024} and Flamingo \cite{Ma2023}, both of which adopt similar masking techniques based on long-term shared secrets and PRFs. Among the two, e-SeaFL demonstrates improved performance and practicality due to its architectural design, which eliminates direct user-to-user interactions and relies instead on a few assisting nodes for mask coordination and recovery. This makes it more suitable for large-scale FL deployments.
\par 
Despite their efficiency, both e-SeaFL and Flamingo share a privacy concern: they do not provide FS and BS under their current single-setup design. In the following section, we analyze e-SeaFL in detail to demonstrate how its use of long-term cryptographic material exposes it to FS and BS vulnerabilities. Since Flamingo employs a similar masking strategy, the same analysis applies, indicating that Flamingo also lacks both FS and BS.

\subsection{Mask Derivation in e-SeaFL}
In e-SeaFL, each user $i$ and assisting node $a \in \mathcal{A}_i$ perform a key agreement protocol ($\mathsf{KA}$) to derive a long-term shared secret or seed in the setup phase:

\[
s_{i,a} = \mathsf{KA}(SK_i, PK_a),
\]

where $SK_i$ and $PK_a$ are the private key of user $i$ and the public key of assisting node $a$, respectively and $\mathcal{A}_i$ is the set of participating assisting nodes. This shared secret $s_{i,a}$ is then used in each training round $t$ as input to a $\mathtt{PRF}$ to generate a round-specific mask value:

\begin{align}
    m_{i,a}^{(t)} &= \mathtt{PRF}(s_{i,a}, t) \label{eq:mask-fun}
\end{align}

The final masked update sent by user $i$ in round $t$ is constructed as:

\[
\tilde{x}_i^{(t)} = x_i^{(t)} + \sum_{a \in \mathcal{A}_i} m_{i,a}^{(t)}.
\]
where $x_i^{(t)}$ is the local model update of the user $i$ for the round $t$. 
\par 
To recover the true global model update $\theta_t$, the server collaborates with assisting nodes to remove the aggregated masks. Each assisting node $a$ computes the sum of its masking contributions across all users, say $\mathcal{U}$, it assists: 

\[
\sum_{i \in \mathcal{U}} m_{i,a}^{(t)}= \sum_{i \in \mathcal{U}} \mathtt{PRF}(s_{i,a}, t) 
\]

After receiving all the aggregated masking contributions from the assisting nodes, the central server removes these masks from the aggregated global model parameters to recover the plaintext model update, as follows:
\[
\theta= \sum_{i\in \mathcal{U}}x_i^{(t)} = \sum_{i\in \mathcal{U}}\tilde{x}_i^{(t)} - \sum_{a\in \mathcal{A}_i}\left(\sum_{i \in \mathcal{U}} m_{i,a}^{(t)}\right)
\]

\subsection{Forward and Backward Secrecy Limitations}
In e-SeaFL, round-specific masking values \( m_{i,a}^{(t)} \) are deterministically derived from long-term shared secrets \( s_{i,a} \) between user \( i \) and assisting node \( a \), as defined in Eq.~\ref{eq:mask-fun}. If only a subset of these secrets is compromised, i.e. \( \mathcal{A}'_i \subset \mathcal{A}_i \), the adversary can compute only partial masking shares. Since at least one share remains unknown, the masked update retains information-theoretic secrecy. This threshold-like structure provides resilience against partial compromise and is a notable strength of e-SeaFL.
\par 
However, this protection collapses entirely if \emph{all} long-term secrets \( s_{i,a} \) for \( a \in \mathcal{A}_i \) are compromised. The adversary can retroactively and proactively compute all masking shares \( m_{i,a}^{(t)} \) across every round \( t \), enabling full reconstruction:


\begin{equation}
x_i^{(t)} = \tilde{x}_i^{(t)} - \sum_{a \in \mathcal{A}_i} m_{i,a}^{(t)}.
\end{equation}

In other words:
\begin{itemize}
    \item \textbf{Forward secrecy is broken:} Compromising the long-term secrets at time~\(T\) enables recovery of all past updates (\(t < T\)).  
    \item \textbf{Backward secrecy is broken:} The same compromise allows prediction of all future masks (\(t > T\)) and recovery of future updates.  
\end{itemize}

\subsection{Practical Implications}
The absence of FS and BS in e-SeaFL (similarly in Flamingo) introduces a critical privacy concern in real-world FL deployments. Devices such as smartphones, edge nodes, and IoT sensors frequently operate in hostile environments--- exposed to malware, physical tampering, and insider threats. Over prolonged training sessions, the probability of compromise increases significantly.
\par 
Once a long-term secret \( s_{i,a} \) is established between a user and an assisting node, it remains static throughout the entire training lifecycle. If compromised, even once, an adversary can deterministically regenerate all masking values \( m_{i,a}^{(t)} \) across every round. This enables full reconstruction of the user’s masked updates. The breach is not limited to a single round; it retroactively exposes past contributions and proactively jeopardizes future ones. While e-SeaFL achieve efficiency through single-setup designs, its reliance on static secrets makes it unsuitable for deployments demanding robust confidentiality. 
\par 
In contrast, our protocol eliminates this risk entirely. By using fresh, per-round randomness and avoiding long-term key dependencies, we ensure that each training round is cryptographically isolated. Even if a device is compromised, the damage is strictly contained to that round. There is no possibility of past and future exposures. This design offers not just theoretical security, but practical resilience without requiring complex key ratcheting, proactive key refreshes, or threshold-based updates.

\section{\(\alpha\)-summation ideal functionality}
\label{sec:appendix-alpha-summation}
As discussed in~\cite{Bell2020} and \cite{Behnia2024}, a summation protocol is referred to as \(\alpha\)-secure if it ensures that each honest participant's input is aggregated with at least \(\alpha \cdot |\mathcal{H}|\) other secret values, where \(\mathcal{H}\) denotes the set of honest parties. We adopt this notion of \(\alpha\)-security in our simulation-based proof framework, consistent with the above works.

\begin{definition}[\(\alpha\)-summation ideal functionality]
Let \( p, n, d \in \mathbb{Z} \), and let \( \alpha \in [0,1] \). Suppose we have a subset \( S \subseteq [|\mathcal{U}|] \), and define the collection \( \mathbf{X}_S = \{ \mathbf{x}_i \}_{i \in S} \), where each vector \( \mathbf{x}_i \in \mathbb{Z}_p \). Let \( \mathcal{P}_S \) represent the set of all partitions of \( S \), and consider a partition \( \{S_1, \ldots, S_k\} \in \mathcal{P}_S \) composed of pairwise disjoint subsets.

The \(\alpha\)-summation ideal functionality is modeled by a function \( \mathcal{F}_{\mathbf{X},\alpha}(\cdot) \), which takes as input the partition \( \{S_j\}_{j=1}^k \) and outputs a collection \( \{\mathbf{s}_j\}_{j=1}^k \), where each \( \mathbf{s}_j \) is computed as follows:

\[
\mathbf{s}_j =
\begin{cases}
\sum\limits_{i \in S_j} \mathbf{x}_i & \text{if } |\mathcal{P}_S| \geq \alpha \cdot |S| \\
\bot & \text{otherwise}
\end{cases}
\quad \text{for all } j \in [1, \ldots, k].
\]
\end{definition}

\section{Security Against Semi-Honest Users Only}
\begin{proof}\label{appendix:Semi-Honest-Users-Only}
In our protocol, users randomly choose their input values (i.e., random secret keys) to mask the local model parameters, similar to a one-time pad. Each user's input values are independent of those of other users. Since the aggregator and intermediate servers are considered honest entities, the combined view of the users in $C_c$ will not reveal any useful information about the input values of users not in $C_c$. Additionally, the honest-but-curious users in $C_c$ only receive the identities of the honest users not in $C_c$. This enables the simulator $\mathsf{Sim_{C}^{\mathbb{U}, t_c, \lambda}}$ to use dummy input values for the honest users not in $C_c$ while keeping the combined views of the users in $C_c$ identical to that of $\mathsf{Real_{C}^{\mathbb{U}, t_c, \lambda}}$.
\end{proof}

\section{Security Against Semi-Honest Intermediate Servers Only}
\begin{proof}\label{appendix:semi-honest-intermediate-servers-only}

In this proof, we consider that the intermediate servers in $C_f$ are dishonest. Since users randomly and independently choose their masking parameters (i.e., random secret keys), the intermediate servers in $C_f$ will not gain any useful information from the public values. Similarly, users send the masked local models (which are independent of each other and similar to the one-time pads) to each intermediate server separately using a secure and authenticated channel. Therefore, the joint views of the intermediate servers in $C_f$ are independent of those of the users and intermediate servers not in $C_f$. Consequently, the simulator $\mathsf{Sim_{C}^{\mathbb{U}, t_f, \lambda}}$ can use dummy input values for the intermediate servers not in $C_f$ while keeping the combined views of the intermediate servers in $C_f$ identical to those in $\mathsf{Real_{C}^{\mathbb{F}, t_f, \lambda}}$.

        
\end{proof}

\section{Security Proof Against Semi-Honest Aggregators and Intermediate Servers}\label{proof:semihonest}
\begin{proof}
    In this proof, we use a standard hybrid argument, which consists of a sequence of hybrid distributions, to construct the simulator $\mathsf{Sim}$ by the subsequent modifications to the random variable $\mathsf{Real}$. The goal is to prove that two subsequent hybrids are computationally indistinguishable to ensure that the distribution of simulator $\mathsf{Sim}$ as a whole is also identical to the real execution $\mathsf{Real}$.
    
    \paragraph*{$\mathsf{Hyb_0}$} In this hybrid, the distribution of the joint views of $C_c$ and $C_f$ is exactly the same as that of $\mathsf{Real}$.

    \paragraph*{$\mathsf{Hyb_1}$} In this hybrid, the simulator $\mathsf{Sim}$ chooses random numbers in $\mathbb{Z}_q$ using a pseudo-random function $\mathrm{PRF}$ as the trained model for the users in $(\mathbb{U}_1\setminus C_c)$. Based on the security of the pseudo-random function, the joint views of the users in $C_c$ is identical to that of $\mathsf{Real}$. Hence, this hybrid is indistinguishable from the previous one. 
    
    \paragraph*{$\mathsf{Hyb_2}$} In this hybrid, similar to the previous one, but for each user $i$ in $(\mathbb{U}_1\setminus C_c)$, the simulator $\mathsf{Sim}$ generates separate one-time-pads to mask $r_i$ for each intermediate server in $(\mathbb{F}\setminus C_f)$ including the aggregator. Based on the security of the one-time-pad, the joint views of the users in $C_c$ and intermediate servers in $C_f$ is identical to that of $\mathsf{Real}$. Hence, this hybrid is indistinguishable from the previous one.
   
    
    \paragraph*{$\mathsf{Hyb_3}$} In this hybrid, for each user $i$ in $(\mathbb{U}_1\setminus C_c)$, the simulator $\mathsf{Sim}$ chooses two random secrets $(k_j, k_{j-1})\in \mathbb{Z}_q$ to mask $r_i$ as $r_i+ k_j- k_{j-1}$ for each intermediate server $j$ in $(\mathbb{F}\setminus C_f)$ including the aggregator. Based on the semantic security (IND-CPA) of the symmetric homomorphic encryption scheme (described in Section \ref{sec:symmetricKeyHomoEncrypt}), the joint views of the users in $C_c$ and intermediate servers in $C_f$ is identical to that of $\mathsf{Real}$. Hence, this hybrid is indistinguishable from the previous one.
     

    \paragraph*{$\mathsf{Hyb_4}$} This hybrid is similar to the previous one, with the only difference being that the simulator \(\mathsf{Sim}\) replaces the random masked values \(r_i\) for each user in \((\mathbb{U}_1 \setminus C_c)\) with values \(x_i\), subject to the constraint:
\begin{align}\label{eq:mask}
    \sum_{i \in \mathbb{U}_1 \setminus C_c} r_i = \sum_{i \in \mathbb{U}_1 \setminus C_c} x_i
\end{align}
The distribution of the masking values \(r_i\) and \(x_i\), when generated using a symmetric homomorphic encryption scheme, is identical subject to the Eq. \ref{eq:mask}.
   \par 
  The last hybrid shows that we can define a simulator $\mathsf{Sim}$ which is computationally indistinguishable from that of the real execution of $\mathsf{Real}$ from the combined views of the honest-but-curious entities. Hence, it completes the proof. 
\end{proof}
\section{Security Against Dishonest Intermediate Servers Only}\label{appendix:dishonest-intermediate-servers-only}
\begin{proof}
    This proof is similar to that of Theorem \ref{theorem:2}, as the intermediate servers in $C_f$ gain no knowledge about $x_{\mathbb{F}\setminus C_f}$ apart from the list of participants. Consequently, the simulator $\mathsf{Sim}$ can assign real inputs to the dishonest intermediate servers while assigning dummy inputs to the remaining intermediate servers, accurately replicating the perspectives of the intermediate servers in $C_f$. Thus, the joint views of the intermediate servers in $C_f$ in the simulation are indistinguishable from those in $\mathsf{Real}$.
\end{proof}
\section{Security Proof Against Dishonest Users, Intermediate Servers and Aggregator}\label{proof:malicious}
\begin{proof}
    We use the standard hybrid argument to carry out the security proof of this theorem. Similar to the Theorem \ref{theorem:3}, we construct a simulator $\mathsf{Sim}$ by the subsequent modifications to the random variable $\mathsf{Real}$. Our goal is to prove that two subsequent hybrids are computationally indistinguishable to ensure that the distribution of simulator $\mathsf{Sim}$ as a whole is also identical to the real execution $\mathsf{Real}$.
    
    \paragraph*{$\mathsf{Hyb_0}$} In this hybrid, the distribution of the joint view of $M_c$ of $\mathsf{Sim}$ is the same as that of $\mathsf{Real}$.

   \paragraph*{$\mathsf{Hyb_1}$} In this hybrid, the simulator $\mathsf{Sim}$ replaces the masking values of the locally trained models with randomly generated numbers of appropriate length. Based on the security of the pseudo-random function, the joint views of the users in $M_c$ are identical to those of $\mathsf{Real}$. Hence, this hybrid is indistinguishable from the previous one.
    
    \paragraph*{$\mathsf{Hyb_2}$} In this hybrid scenario, the simulator $\mathsf{Sim}$ replaces the masked locally trained model $x_i$ of each honest user $u_i$ with a randomly generated number, while it replaces it with $0$ for the adversaries. Based on the security of the pseudo-random function, the joint views of the users in $M_c$ are identical to those of $\mathsf{Real}$. Hence, this hybrid is indistinguishable from the previous one.
  
    \paragraph*{$\mathsf{Hyb_4}$} In this hybrid, the simulator $\mathsf{Sim}$ uses the symmetric homomorphic encryption with key negation method, as described in Section \ref{sec:symmetricKeyHomoEncrypt}, to replace each mask of the locally trained model for the intermediate servers in the $\mathbb{F}'$. Based on the semantic security (IND-CPA) of the symmetric homomorphic encryption scheme, the joint views of the users in $M_c$ are identical to those of $\mathsf{Real}$. Hence, this hybrid is indistinguishable from the previous one. 

    
    \paragraph*{$\mathsf{Hyb_5}$} In this hybrid, the $\mathsf{SIM}$ aborts if $M_c$ provides with incorrect signatures $\sigma^{i, j}_t$ (or $\sigma^{i, \mathsf{Agg}}_t$) for each user $u_i$ in $(\mathbb{U} \setminus C_c)$. This hybrid is indistinguishable from the previous one based on the security assumption of the signature mechanism.

    \paragraph*{$\mathsf{Hyb_6}$} In this hybrid, the simulator $\mathsf{SIM}$ fetches $\mathrm{I}$ list of users from the aggregator and aborts if list has any invalid user. This hybrid is identical to the previous one. 
    
    \paragraph*{$\mathsf{Hyb_7}$} In this hybrid, the simulator $\mathsf{SIM}$ queries to the Random Oracle $\mathcal{O}$ for the values $w_i$ for the set $(\mathrm{I}\setminus C_c)$ with respect to $\sum_{u_i\in \mathcal{Q}\setminus C_c} w_i= \sum_{u_i\in \mathcal{Q}\setminus C_c}x_i$ instead of receiving the inputs of the honest users and masks it with a random number. For the adversaries, the $\mathsf{Sim}$ chooses random numbers (masks are set to be $0$ here). The Random Oracle $\mathcal{O}$ will not abort and does not modify the joint view of $M_c$ if at least more than $2$ participants are honest due to the properties of the random oracle and pseudo-random function. Thus, this hybrid is indistinguishable from the previous one.

    \paragraph*{$\mathsf{Hyb_8}$} In this hybrid, the $\mathsf{Sim}$ aborts if $M_c$ provides with incorrect signatures $\sigma^{\mathrm{F}}$ for each intermediate server $f_i$ in $(\mathbb{F} \setminus C_f)$. This hybrid is indistinguishable from the previous one based on the security assumption of the signature mechanism.

    \paragraph*{$\mathsf{Hyb_9}$} In this hybrid, the simulator $\mathsf{SIM}$ again queries to the Random Oracle $\mathcal{O}$ for the values $\{w^j_i\}_{\forall j\in \mathbb{F}\setminus C_f}$ for the set $(\mathrm{I}\setminus C_c)$ with respect to $\sum_{u_i\in \mathcal{Q}\setminus C_c} w'_i= \sum_{u_i\in \mathcal{Q}\setminus C_c}x'_i$ instead of receiving the inputs of the honest users and honest intermediate servers. The Random Oracle $\mathcal{O}$ will not abort and does not modify the joint views of $M_c$ if at least more than $2$ users are honest and $1$ intermediate server is honest. Thus, this hybrid is indistinguishable from the previous one.
    

    \paragraph*{$\mathsf{Hyb_{10}}$} In this hybrid, the simulator \(\mathsf{Sim}\) selects a random secret from \(\mathbb{Z}_q\) to mask the aggregated model, as in the previous hybrid, and computes a message authentication code (MAC) over the aggregated model using a secret key unknown to the adversary. Since the random masking value is never reused, it effectively functions as a one-time pad. Furthermore, the MAC, assumed to be pseudorandom under the PRF assumption, reveals no information to the adversary. Consequently, this hybrid is computationally indistinguishable from the previous one.

\par 
It can be observed from the last hybrid that the adversary's view remains unaffected which proves the indistinguishable property between the different hybrids. Furthermore, this hybrid does not take input from honest parties. This completes the proof.
\end{proof}

\end{document}